\documentclass[pra,aps,reprint,superscriptaddress,longbibliography,floatfix]{revtex4-1}
\usepackage{times}
\usepackage{epsfig}
\usepackage{amsfonts}
\usepackage{amsmath}
\usepackage{amssymb}
\usepackage{amsthm}
\usepackage{color}
\usepackage{multirow}
\usepackage[normalem]{ulem}
\newcommand{\stkout}[1]{\ifmmode\text{\sout{\ensuremath{#1}}}\else\sout{#1}\fi}
\usepackage{latexsym}
\usepackage{mathrsfs}
\usepackage{natbib}
\usepackage{verbatim}
\usepackage[T1]{fontenc}

\DeclareMathOperator{\Tr}{tr}

\newcommand{\ket}[1]{|#1\rangle}
\newcommand{\bra}[1]{\langle#1|}
\newcommand{\braket}[2]{\langle#1|#2\rangle}
\newcommand{\bracket}[3]{\langle#1|#2|#3\rangle}
\newcommand{\ketbra}[2]{|#1\rangle\!\langle#2|}

\newcommand{\id}{\openone}
\newcommand{\mH}{\mathcal{H}}
\newcommand{\mE}{\mathcal{E}}
\newcommand{\mS}{\mathcal{S}}
\newcommand{\tr}[1]{\Tr\left[#1\right]}
\newcommand{\trr}[2]{\Tr_{#1}\left[#2\right]}

\newtheorem{theorem}{Result}
\newtheorem{conjecture}{Conjecture}

\newtheorem{lemma}[theorem]{Lemma}

\usepackage[colorlinks=true,linkcolor=blue,citecolor=magenta,urlcolor=blue]{hyperref}

\begin{document}

%%%%%%%%%%%%%%%%%%%%%%%%%%%%%%%%%%%%%%%%%%%%%%%%%%%%%%%%%%%%%%%%%%%

\title{Semi-device-independent characterisation of multipartite entangled states and measurements}

%%%%%%%%%%%%%%%%%%%%%%%%%%%%%%%%%%%%%%%%%%%%%%%%%%%%%%%%%%%%%%%%%%%

\author{Armin Tavakoli}
\affiliation{D\'epartement de Physique Appliqu\'ee, Universit\'e de Gen\`eve, CH-1211 Gen\`eve, Switzerland}

\author{Alastair A.\ Abbott}
\affiliation{Univ.\ Grenoble Alpes, CNRS, Grenoble INP, Institut N\'eel, 38000 Grenoble, France}

\author{Marc-Olivier Renou}
\affiliation{D\'epartement de Physique Appliqu\'ee, Universit\'e de Gen\`eve, CH-1211 Gen\`eve, Switzerland}

\author{Nicolas Gisin}
\affiliation{D\'epartement de Physique Appliqu\'ee, Universit\'e de Gen\`eve, CH-1211 Gen\`eve, Switzerland}

\author{Nicolas Brunner}
\affiliation{D\'epartement de Physique Appliqu\'ee, Universit\'e de Gen\`eve, CH-1211 Gen\`eve, Switzerland}

%%%%%%%%%%%%%%%%%%%%%%%%%%%%%%%%%%%%%%%%%%%%%%%%%%%%%%%%%%%%%%%%%%%

\begin{abstract}

The semi-device-independent framework allows one to draw conclusions about properties of an unknown quantum system under weak assumptions. 
Here we present a semi-device-independent scheme for the characterisation of multipartite entanglement based around a game played by several isolated parties whose devices are uncharacterised beyond an assumption about the dimension of their Hilbert spaces.
Our scheme can certify that an $n$-partite high-dimensional quantum state features genuine multipartite entanglement. Moreover, the scheme can certify that a joint measurement on $n$ subsystems is entangled, and provides a lower bound on the number of entangled measurement operators. 
These tests are strongly robust to noise, and even optimal for certain classes of states and measurements, as we demonstrate with illustrative examples. Notably, our scheme allows for the certification of many entangled states admitting a local model, which therefore cannot violate any Bell inequality.
\end{abstract}

%%%%%%%%%%%%%%%%%%%%%%%%%%%%%%%%%%%%%%%%%%%%%%%%%%%%%%%%%%%%%%%%%%%

\maketitle

\emph{Introduction.---}%
Entanglement represents a central feature of quantum theory and a key resource for quantum information processing~\cite{H09}. Therefore, the task of characterising entanglement experimentally is of fundamental significance. In particular, the development of quantum networks, multiparty cryptography, quantum metrology, and quantum computing necessitate certification methods tailored to multipartite entangled states, as well as to entangled joint measurements. 

Standard methods for the certification of multipartite entanglement rely on entanglement witnesses~\cite{GT09}, as quantum tomography quickly becomes infeasible as the number of subsystems increases. 
Entanglement witnesses can also be used for the particularly important sub-case of certifying genuine multipartite entanglement (GME), the strongest form of multipartite entanglement, where all subsystems are genuinely entangled together; see, e.g., Refs.~\cite{H09,GT09,ES14}. 
In practice, the main drawback of entanglement witnesses is that they crucially rely on the correct calibration of the measurement devices, as a set of specific observables must be measured. 
Importantly, even small alignment errors can have undesirable consequences, e.g.\ leading to false positives~\cite{RFD12}, and it is generally cumbersome to estimate these errors and take them into account rigorously.

This motivates the developments of certification methods that require minimal assumptions on the measurement devices, and in particular do not rely on their detailed characterisation. 
This is the spirit of the device-independent (DI) approach to entanglement characterisation, which also leads to interesting possibilities for quantum information processing~\cite{Acin07,Colbeck,Pironio10}. 
The main idea consists in using Bell inequalities, given that a violation of such an inequality necessarily implies the presence of entanglement in the state (even without any knowledge about the measuring devices).
Moreover, GME can also be detected via Bell-like inequalities~\cite{Collins,Seevinck,BG11, PV11, BB12, MB13, MR16}. 
Experimentally, however, this approach is very demanding as high visibilities are typically required. 
More generally, a broad range of entangled states (including many GME states~\cite{Toth,Remik,Bowles16}) cannot, in fact, violate any Bell inequality~\footnote{Note that nonlocality can nevertheless be activated in some more sophisticated Bell scenarios involving, e.g., sequential measurements~\cite{Popescu} or processing of multiple copies~\cite{Palazuelos}. Furthermore, with the help of additional sources of perfect singlets, any entangled bipartite state can be certified~\cite{BS18}.} as they admit local hidden variable models~\cite{W89,Augusiak_review}.
Finally, although Bell inequalities can in principle be used for the certification of entangled joint measurements~\cite{RH11}, no practical scheme has been reported thus far.

This motivates the exploration of partially DI scenarios, in between the fully DI case of Bell inequalities and the device-dependent case of entanglement witnesses. Here, only weak assumptions about the devices are typically made. One possibility is to consider that a subset of parties perform well-characterised measurements, while the others are uncharacterised~\cite{Reid,Paul,Diamanti}. 
Another option is to consider Bell experiments with quantum inputs, leading to the so-called measurement device-independent characterisation of entanglement~\cite{B12,Branciard}. 
While experimental demonstrations have been reported, both of these approaches have the drawback of requiring certain parts of the experiment to be fully characterised.

In the present work, we follow a different approach for entanglement characterisation. 
Specifically, we will assume only an upper bound on the Hilbert space dimension of the subsystems of interest, but require no detailed characterisation of any of the devices. 
Roughly speaking, this assumption means that all the relevant degrees of freedom are described in a Hilbert space of given dimension~\cite{PB11, LP12, TH15}, and that other potential side-channels can be neglected. 
This scenario, usually referred to as the semi-DI (SDI) setting, has been considered for the characterisation of entanglement in the simplest setting of two-qubit states~\cite{LV11,Koon}, as well as the detection of two-qubit entangled measurements~\cite{VN11,BV14}.

Here, we present a versatile scheme for characterising both multipartite entangled states and entangled measurements in a semi-DI setting. 
Our scheme allows one to simultaneously certify that (I) an $n$-partite quantum state (of arbitrary local dimension) is GME, and that (II) a measurement performed on the $n$ subsystems is entangled. 
Furthermore, we obtain a finer characterisation for the measurement, namely a lower bound on the number of entangled measurement operators. 
In general our scheme is strongly robust to noise, and even optimal in certain cases. 
It certifies all noisy qubit Greenberger-Horne-Zeilinger (GHZ) states that are GME, and, in the bipartite case, all entangled isotropic states of arbitrary dimension. 
Other classes of GME states, e.g.\ Dicke states, can also be certified, although not optimally. 
For the case of entangled measurements, we give two illustrative examples. 
In particular, we optimally certify the presence of entanglement in a noisy Bell-state measurement. 
Finally, we conclude with a list of open questions.

\emph{Scenario.---}%
The scenario we consider consists of a state being initially prepared, then transformed by several parties, and finally measured. 
Since we operate within the SDI framework no assumptions are made on the internal workings of any of these parties' devices, other than a bound on their local Hilbert space dimensions, see Fig.~\ref{fig}. 
Throughout the experiment the parties cannot communicate amongst themselves; one may consider, e.g., that they are spacelike separated as is usual in the DI framework.

Let an uncharacterised source distribute a state $\rho$ of arbitrary dimension between $n$ parties $A_1,\ldots,A_n$, each of which receives a subsystem. 
Each party $A_k$, for $k=1,\ldots, n$, receives uniformly random inputs $x_k,y_k\in\{0,\ldots,d-1\}$. 
Subsequently, they perform local transformations $\mathcal{T}^{(k)}_{x_k y_k}$ (which may be any completely positive trace-preserving (CPTP) map) which map their local states into a $d$-dimensional state. The transformed state is sent to a final party denoted by $B$ who performs a (possibly joint) measurement $\{M_\mathbf{b}\}_\mathbf{b}$ (i.e., a positive-operator valued measure (POVM) with $M_\mathbf{b}\ge 0$ and $\sum_\mathbf{b} M_\mathbf{b} = \id$) which produces an outcome string $\mathbf{b}=b_1\ldots b_n\in\{0,\ldots,d-1\}^{n}$ (see Fig.~\ref{fig}). 
\begin{figure}
	\centering
	\includegraphics[width=0.9\columnwidth]{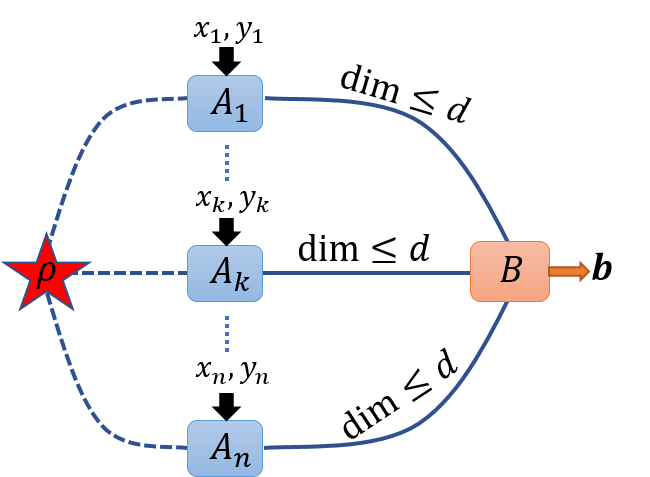}
	\caption{An $n$-partite state is distributed between parties $A_1,\ldots, A_n$ who perform transformations on their local systems depending on their random inputs $(x_k,y_k)$. Each party sends their transformed $d$-dimensional system to $B$ who performs a measurement and obtains an outcome $\mathbf{b}$. 
		}\label{fig}
\end{figure}
The experiment gives rise to a probability distribution $P(\mathbf{b}|\mathbf{x},\mathbf{y})$, where $\mathbf{x}=x_1\ldots x_n$ and $\mathbf{y}=y_1\ldots y_n$, given by 
\begin{equation}
P(\mathbf{b}|\mathbf{x},\mathbf{y})= \Tr\left[\left(\bigotimes_{k=1}^{n} \mathcal{T}^{(k)}_{x_k y_k}\right)\!\!\left[\rho\right] \cdot M_{\mathbf{b}}\right].
\end{equation}
The goal of the task (or game) we consider is for the parties to cooperate so that $B$'s output satisfies the conditions
\begin{equation}\label{wincond}
\hspace{2mm} b_1=\sum_{i=1}^{n}x_i\equiv C_1(x) \hspace{2mm} \text{and}  \hspace{2mm} b_k= y_k-y_1\equiv C_k(y),
\end{equation} 
for $k=2,\ldots,n$, where all quantities are computed modulo $d$.
Compactly, we write $C(\mathbf{x},\mathbf{y})$ for the (unique) string $\mathbf{b}$ satisfying all the above conditions. Note that these conditions are not totally symmetric (except when $n=2$), but they will nonetheless prove useful in certifying entanglement. 
Given a strategy leading to a probability distribution $P(\mathbf{b}|\mathbf{x},\mathbf{y})$, the probability of winning the task (or if a win is rewarded with a point, the average score) is thus given by
\begin{equation}\label{test}
\mathcal{A}_{n,d}=\frac{1}{d^{2n}}\sum_{\mathbf{x},\mathbf{y}} P\left(\mathbf{b}=C(\mathbf{x},\mathbf{y})\,|\, \mathbf{x},\mathbf{y}\right).
\end{equation}
We now show how, from the value of an observed average score $\mathcal{A}_{n,d}$, one can make inferences about the entanglement of the state $\rho$ and the measurement $\{M_\mathbf{b}\}_\mathbf{b}$.

\emph{Characterising entangled states.---}%
We first consider certifying the GME of the shared state. 
A state is said to be GME if it is not biseparable, i.e., if it cannot be written in the form $\rho=\sum_{S}\sum_{i} p_{S,i}\rho_i^{S}\otimes \rho_i^{\bar{S}}$, for any possible bipartition $\{S,\bar{S}\}$ of the subsystems $\{1,\ldots,n\}$, where $\sum_{S,i}p_{S,i}=1$ and $p_{S,i}\geq 0$. 

We now show that the value of $\mathcal{A}_{n,d}$ can be nontrivially upper bounded for any $n$-partite biseparable state. This will allow us to certify GME since, as we will see later, the bound is violated by many GME states of interest.
\begin{theorem}\label{res1}
	Let $\rho$ be a state of $n$ subsystems. For any measurement $\{M_\mathbf{b}\}_\mathbf{b}$ and any transformations $\{\mathcal{T}^{(k)}_{x_ky_k}\}_k$, it holds that
	\begin{equation}\label{state}
	\rho \emph{ is biseparable} \implies  \mathcal{A}_{n,d}\leq 1/d.
	\end{equation}
	Hence, whenever $\mathcal{A}_{n,d}>1/d$, $\rho$ is certified to be GME.	Moreover, this inequality is tight and the bound can be saturated with fully separable states. 
\end{theorem}
\begin{proof}
The full details of the proof are given in Appendix~\ref{A}.
In order to prove the upper bound in~\eqref{state}, we consider a relaxed SDI task in which any distribution $P(\mathbf{b}| \mathbf{x},\mathbf{y})$ obtainable in the original task is also possible in the relaxed setting, but not vice versa. The relaxation is chosen so that the average score $\mathcal{A}_{n,d}$ can easily be upper-bounded. By construction, the upper bound obtained in the relaxed scenario is also valid for the original task.

To see that the bound is tight, we give a strategy that utilises only product states and saturates the bound~\eqref{state}. Let $A_k$ (for $k=1,\ldots, n$) send $y_k$ to party $B$. With this information, $B$ can output $b_i=y_i-y_1$, satisfying condition $C_i$, for $i=2,\ldots, n$. However, this strategy forces $B$ to guess $b_1$ in order to satisfy condition $C_1$. Any such guess succeeds, on average, with probability $1/d$, thus saturating the bound. 
\end{proof}

In order to show the relevance of the relation in Eq.~\eqref{state}, we show that it can be violated by GME states. In particular, we first consider the largest achievable value of $\mathcal{A}_{n,d}$. It turns out that the algebraically maximal value, i.e., $\mathcal{A}_{n,d}=1$, can be achieved for all $n$ and $d$ via the following strategy. A GME state of $n$ subsystems of local dimension $d$, namely the generalised GHZ state,
\begin{equation} 
\ket{\text{GHZ}_{n,d}}=\frac{1}{\sqrt{d}}\sum_{i=0}^{d-1}\ket{i}^{\otimes n} \,
\end{equation}
is distributed among the parties $A_1,\ldots, A_n$. Each party then performs the unitary transformation $U^{A_k}_{x_ky_k}=Z^{x_k}X^{y_k}$ for $k=1,\ldots, n$, where
\begin{equation}\label{Unitaries}
Z=\sum_{j=0}^{d-1} e^{2i\pi j /d}\ketbra{j}{j} \, , \, \quad  X=\sum_{j=0}^{d-1}\ketbra{j+1}{j} 
\end{equation}
are the usual clock and shift operators. 
Finally, $B$ performs a joint projective measurement in the basis of generalised GHZ states given by
\begin{equation}\label{BellMeas}
	\ket{M_{b}}=Z^{b_1}\otimes X^{b_2}\otimes \cdots\otimes X^{b_n}\ket{\text{GHZ}_{n,d}}.
\end{equation}
Note that in the simplest case of two qubits ($n=d=2$), the four unitaries are simply the three Pauli matrices and the identity matrix, while the measurement is the Bell-state measurement~\cite{BB93}. 

Let us now consider noisy GHZ states, i.e., mixtures of GHZ states with white noise: $\rho_{n,d}^{\text{GHZ}}(v)=v\,\ket{\text{GHZ}_{n,d}}\bra{\text{GHZ}_{n,d}}+(1-v)\id/d^n$, where $v\in[0,1]$ is the visibility of the state. 
In the strategy given above, for the white noise state one has $\mathcal{A}_{n,d}(\id/d^n)=1/d^n$. 
Hence, from the linearity of $\mathcal{A}_{n,d}$ in $\rho$, it follows that a violation of Eq.~\eqref{state} is obtained whenever $v+(1-v)/d^n>1/d$, that is, when $v>(d^{n-1}-1)/(d^n-1)$. We discuss the implications of this result in three separate cases of interest. 

(I) For two $d$-dimensional systems ($n=2$), the criterion is $v>1/(d+1)$ which is precisely the condition for the entanglement of $\rho_{2,d}^{\text{GHZ}}(v)$~\cite{HHH99}. 
Hence, every entangled isotropic bipartite state is certified by our protocol. 
Interestingly, such certification is impossible using Bell inequalities; for instance we note that the state $\rho_{2,2}^{\text{GHZ}}(v)$ has a local hidden variable model (for projective measurements) when $v<0.6829$~\cite{HQ17}, and will therefore not violate any (known or unknown) Bell inequality. 
For large $d$, such models are known for $v\leq \log d/d$~\cite{Almeida},
and, in the limit, known facet Bell inequalities allow for the certification of $\rho_{2,d}^{\text{GHZ}}(v)$ when $v>0.67$~\cite{CGLMP02}.
In contrast, as $d\to\infty$ our protocol can certify entanglement for arbitrary small $v$.

(II) In the case of a system of many qubits ($d=2$), our visibility criterion coincides with the condition for   $\rho_{n,2}^{\text{GHZ}}(v)$ to be GME \cite{RH12}. Hence our scheme can certify all noisy qubit GHZ states that are GME. Again, this would not be possible using Bell inequalities, as (for instance) the state $\rho_{3,2}^{\text{GHZ}}(v=1/2)$ is GME but admits a biseparable model reproducing all correlations from projective measurements~\cite{BG11}. 
Furthermore, the GME of $\rho_{3,2}^{\text{GHZ}}(v)$ is known to be certifiable via a Bell inequality (in the limit of many measurements) when $v> 0.64$ \cite{PV11}, whereas our criterion reads $v>3/7$. 

(III) When considering many high-dimensional systems ($n>2$ and $d>2$), our setup is no longer optimal since there exist generalised GHZ states that are GME below the critical visibility of our scheme~\cite{CH17}. 
Nevertheless, choosing, for instance, $n=d=3$, the criterion is $v>4/13$ which substantially outperforms the certification obtainable via known Bell inequalities, which is possible when $v>0.81$~\cite{BG11,MR16}.

More generally, we derive a lower bound on the maximal value of $\mathcal{A}_{n,d}$ achievable for an arbitrary state $\rho$. 
To this end, consider the following strategy. 
Let $B$ perform the measurement given by~\eqref{BellMeas}, and let each party $A_k$ perform the transformation $\mathcal{T}^{(k)}_{x_ky_k}[\rho]=(U^{A_k}_{x_ky_k})\,\Lambda_k[\rho]\,(U^{A_k}_{x_ky_k})^\dagger$, where $U^{A_k}_{x_ky_k}=Z^{x_k}X^{y_k}$ and the $\Lambda_k$, for $k=1,\dots,n$, are CPTP maps.
Evaluating the score with this strategy and optimising over the CPTP maps $\Lambda_1,\dots,\Lambda_n$ straightforwardly leads to
\begin{align}\label{lowerbound}
 \mathcal{A}_{n,d}(\rho) = \mathrm{EGF}_{n,d}(\rho),
\end{align}
where we have defined the quantity $\mathrm{EGF}_{n,d}(\rho)=\max_{\Lambda_1,\dots,\Lambda_n} \Tr\left[\left(\bigotimes_{k=1}^n \Lambda_i\right)\!\big[\rho\right]\!\cdot\!\ketbra{\text{GHZ}_{n,d}}{\text{GHZ}_{n,d}} \big]$.
If one instead maximises only over unitary maps $\Lambda_k[\rho]=V_k\rho V_k^\dagger$ one obtains $A_{n,d}=\mathrm{GF}_{n,d}(\rho)$, where 
$\mathrm{GF}_{n,d}(\rho)=\max_{V_1,\dots,V_n} \Tr\!\big[\left(\bigotimes_{k=1}^n V_i\right)\rho\left(\bigotimes_{k=1}^n V_i\right)^\dagger  \!\ketbra{\text{GHZ}_{n,d}}{\text{GHZ}_{n,d}}\big]$ is the \emph{GHZ fraction}~\cite{Xu16}, a multipartite generalisation of the singlet fraction~\cite{HHH99}.
$\mathrm{EGF}_{n,d}(\rho)$ can then be seen as the ``extractable'' GHZ fraction, an important generalisation since, even in the bipartite case, local CPTP maps can increase the singlet fraction of an entangled state~\cite{badziag}.

In order to determine whether one can obtain a better score than that given by Eq.~\eqref{lowerbound} by considering arbitrary transformations and measurements, we conducted extensive numerical tests.  
Focusing on the cases $(n,d)\in\{(2,2),(2,3),(3,2)\}$ and optimising numerically $\mathcal{A}_{n,d}$ starting from randomly chosen transformations and measurements, we were unable to obtain a better score than $\mathrm{EGF}_{n,d}(\rho)$.
We note also that, when restricted to unitary transformations, we were similarly unable to obtain a score larger than $\mathrm{GF}_{n,d}(\rho)$.
Motivated by this numerical evidence, we make the following conjecture: 
\begin{conjecture}
	Let $\rho$ be an $n$-partite state of local dimension $d$.  
	Then the maximal value of $\mathcal{A}_{n,d}$ achievable for any measurement $\{M_\mathbf{b}\}_\mathbf{b}$ and transformations $\{\mathcal{T}^{(k)}_{x_ky_k}\}_k$ is $\mathrm{EGF}_{n,d}(\rho)$, i.e.,
	\begin{equation}
	\max_{\{\mathcal{T}^{(k)}_{x_ky_k}\},\{M_\mathbf{b}\}_\mathbf{b}} \mathcal{A}_{n,d}(\rho)={\mathrm{EGF}}_{n,d}(\rho).
	\end{equation}
\end{conjecture}
Proving this conjecture would be particularly interesting in the bipartite case, as violation of our witness would then certify an extractable singlet fraction greater than $1/d$, which implies that maximal entanglement can be distilled from the state $\rho$~\cite{HHH99}. 

We conclude this part, by discussing other classes of GME states, that are qualitatively inequivalent to GHZ states. 
Let us first consider noisy W states of three qubits, i.e., $W(v)=v\,\ketbra{W}{W}+(1-v)\id/8$, where $\ket{W}=\tfrac{1}{\sqrt{3}}\big(\ket{001}+\ket{010}+\ket{100}\big)$. 
Numerical optimisation gives a (seemingly optimal) strategy obtaining $\mathcal{A}_{3,2}=\frac{1}{8}(1+5v)$.
This implies that our scheme certifies the GME of $W(v)$ for $v>3/5$. 
This is relatively close to the optimal visibility for GME of $v>0.48$~\cite{JM11}. 
In comparison, known DI schemes based on Bell inequalities would require $v>0.72$~\cite{BG11}. 

As a second illustrative example, we consider a noisy four-qubit Dicke state $D(v)=v\,\ketbra{D}{D}+(1-v)\id/16$, where $\ket{D}=\frac{1}{\sqrt{6}}\big(\ket{0011}+\ket{0101}+\ket{0110}+\ket{1001}+\ket{1010}+\ket{110}\big)$. 
Numerically, the best strategy we find certifies the GME of $D(v)$ for $v>7/11\approx 0.64$, while it is known to be GME for $v>0.46$~\cite{JM11}.

\emph{Characterising entangled measurements.---}%
Next, we consider characterising the entanglement of the joint measurement performed by $B$. 
A measurement $\{M_\mathbf{b}\}_\mathbf{b}$ is said to be entangled if at least one measurement operator $M_\mathbf{b}$ does not have a fully separable decomposition $M_\mathbf{b}=\sum_i M_{\mathbf{b},i}$ where $M_{\mathbf{b},i}\ge 0$ and each $M_{\mathbf{b},i}$ has the tensor product form $M_{\mathbf{b},i}=\bigotimes_{k=1}^n M_{\mathbf{b},i}^{(k)}$. 
We will see that the separability of $\{M_\mathbf{b}\}_\mathbf{b}$ imposes a nontrivial bound on $\mathcal{A}_{n,d}$ which will allow us to certify the entanglement of the measurement used. 
However, it is sufficient for a single measurement operator to be entangled in order for the measurement to be entangled. 
This qualitative property rules out a classical description, but reveals little about the extent to which entanglement is present in the measurement. 
Interestingly, we can go a step further and show that the value of $\mathcal{A}_{n,d}$ implies a bound on the minimum number of the $d^n$ measurement operators that are entangled. 
This provides a much finer characterisation of the joint entangled measurement performed by $B$.

\begin{theorem}\label{res2}
Let $\{M_\mathbf{b}\}_\mathbf{b}$ be a joint measurement of an $n$-partite system of local dimension $d$ with at least $k\in\{0,\ldots,d^n\}$ fully separable measurement operators. 
For any $n$-partite state $\rho$ and any transformations $\{\mathcal{T}^{(k)}_{x_ky_k}\}_k$ it holds that
\begin{equation}\label{measurement}
\substack{\emph{At least $k$ separable}\\ \emph{measurement operators}} \implies \mathcal{A}_{n,d}\leq \frac{1}{d^n}\left(d^n-k+\frac{k}{d}\right).
\end{equation}
Hence, a violation of this inequality for a particular $k$ implies that at least $d^n-k+1$ measurement operators are entangled.
\end{theorem}
\begin{proof}
	The proof is given in Appendix~\ref{B}.
\end{proof}

In the extremal case of a fully separable measurement ($k=d^n$), the bound~\eqref{measurement} reduces to $\mathcal{A}_{n,d}\leq 1/d$, and observing a violation of this bound certifies the entanglement of the measurement.
In the other extreme of witnessing a measurement whose operators are all entangled (i.e., violating Eq.~\eqref{measurement} for $k=1$), the bound~\eqref{measurement} is $\mathcal{A}_{n,d}\leq 1-(d-1)/d^{n+1}$, which is nontrivial for all $n$ and $d$.
In Appendix~\ref{D} we give some partial results on the tightness of Eq.~\eqref{measurement}.

To demonstrate the usefulness of Result~\ref{res2}, we now discuss two illustrative examples in the bipartite case. 
First we consider a noisy version of a generalised Bell-state measurement (for arbitrary $d$), which, we
recall, is a joint measurement of two $d$-level systems given by the projection onto a basis of $d^2$ maximally entangled states $\{\ket{M_{b_1b_2}}\}_{b_1b_2}$. 
To this end, we consider the measurement operators $M_{b_1b_2}(v)=v\, \ketbra{M_{b_1b_2}}{M_{b_1b_2}}+(1-v)\id/d^2$. 
Note that each of these measurement operators is equivalent (up to local unitaries) to an isotropic state $\rho_{2,d}^{\text{GHZ}}(v)$, and is thus entangled precisely when $v>1/(d+1)$. 

To certify the entanglement of this measurement, consider again the strategy used previously in which the parties share a maximally entangled state of two $d$-level systems, and perform the unitary transformations $U^{A_k}_{x_ky_k}=Z^{x_k}X^{y_k}$. 
Note that the score obtained here is the same as when considering a shared noisy isotropic state $\rho_{2,d}^{\text{GHZ}}(v)$, combined with an ideal Bell-state measurement. 
It thus follows from the previous results that we can certify the entanglement of the measurement whenever $v>1/(d+1)$. 
Hence, we certify entanglement whenever the level of noise is low enough to keep the measurement operators entangled. 
Nevertheless, note that in this case our witness certifies only that at least one measurement operator is entangled, while all the measurement operators are in fact entangled. 
In order to certify the entanglement of all $d^2$ measurement operators, a much higher visibility of $v>(d^2 + d -1)/( d^2+d)$ is required. 
In the simplest case of qubits ($d=2$) this requirement is $v>5/6$.

The second example we consider is for the case $n=d=2$, where we look at a two-qubit Bell-state measurement subjected to coloured noise. 
Defining the usual Bell basis $\ket{\phi^{\pm}}=\frac{1}{\sqrt{2}}(\ket{00}\pm\ket{11})$, $\ket{\psi^{\pm}}=\frac{1}{\sqrt{2}}( \ket{01}\pm\ket{10})$, and $\Phi^{\pm}=\ketbra{\phi^{\pm}}{\phi^{\pm}}$, $\Psi^{\pm}=\ketbra{\psi^{\pm}}{\psi^{\pm}}$, consider the measurement given by the operators
\begin{align}
& E_{00}=v\Phi^++\frac{1-v}{4}\left(\Phi^++2\Phi^-+\Psi^+\right)\\
& E_{01}=v\Psi^++\frac{1-v}{4}\left(\Phi^++\Phi^-+2\Psi^+\right)\\
& E_{10}=v\Phi^-+\frac{1-v}{4}\left(2\Phi^++\Phi^-+\Psi^+\right)\\
& E_{11}=\Psi^-,
\end{align}
for some visibility $v\in[0,1]$. 
When $1/3< v\leq 1$, all four operators are entangled; when $0<v\leq 1/3 $, only $E_{01}$ and $E_{11}$ are entangled; and when $v=0$, only $E_{11}$ is entangled. 

Considering the same strategy as in the previous example (i.e., sharing a maximally entangled state and applying the transformations $U^{A_k}_{x_ky_k}$), we find $\mathcal{A}_{2,2}=(1+v)/2$. 
By virtue of Eq.~\eqref{measurement}, this certifies the presence of four entangled measurement operators when $v>3/4$, at least three when $v>1/2$, at least two when $v>1/4$, and at least one when $v>0$.

\emph{Conclusion.---}%
We presented a scheme for the semi-device-independent characterisation of multipartite entangled states and measurements. 
In particular, our scheme can both certify that states are GME, and provide a lower bound on the number of measurement operators that are entangled. 
We illustrated the relevance of our scheme in various case studies, showing strong robustness to noise. 
In particular, we showed that at the price of a dimensional bound, one can overcome fundamental limitations of entanglement certification using Bell inequalities, as our scheme can certify many entangled states admitting a local model. 

We conclude with some relevant open questions. 
(I) Does Eq.~\eqref{measurement} hold also for biseparable (rather than fully separable) measurement operators, thereby allowing the number of GME measurement operators to be bounded?
(II) Can a more sophisticated semi-DI scheme certify the entanglement of all GME states (not just particular classes of states, as shown here)? 
(III) Can our scheme be used as a dimension witness for both states and measurements (e.g., obtaining a score $\mathcal{A}_{n,d}=1$ seems to require a state and measurement of dimension at least $d^n$)?
(IV) Note that a score $\mathcal{A}_{n,d}=1$ implies that all the relations \eqref{wincond} are satisfied. This makes the scheme a good candidate for multiparty cryptographic tasks (e.g., secret sharing of classical data with quantum resources). Exploring this possibility would be interesting. 
(V) Can the scheme be used for self-testing states, measurements, and transformations in prepare-and-measure experiments~\cite{TK18}? 

\emph{Acknowledgements.---}%
We thank Cyril Branciard for discussions. This work was supported by the Swiss national science foundation (Starting grant DIAQ, NCCR-QSIT) and the French National Research Agency (Retour Post-Doctorants programme grant ANR-13-PDOC-0026).

\bibliography{SDI_GME}

\begin{thebibliography}{47}%
\makeatletter
\providecommand \@ifxundefined [1]{%
 \@ifx{#1\undefined}
}%
\providecommand \@ifnum [1]{%
 \ifnum #1\expandafter \@firstoftwo
 \else \expandafter \@secondoftwo
 \fi
}%
\providecommand \@ifx [1]{%
 \ifx #1\expandafter \@firstoftwo
 \else \expandafter \@secondoftwo
 \fi
}%
\providecommand \natexlab [1]{#1}%
\providecommand \enquote  [1]{``#1''}%
\providecommand \bibnamefont  [1]{#1}%
\providecommand \bibfnamefont [1]{#1}%
\providecommand \citenamefont [1]{#1}%
\providecommand \href@noop [0]{\@secondoftwo}%
\providecommand \href [0]{\begingroup \@sanitize@url \@href}%
\providecommand \@href[1]{\@@startlink{#1}\@@href}%
\providecommand \@@href[1]{\endgroup#1\@@endlink}%
\providecommand \@sanitize@url [0]{\catcode `\\12\catcode `\$12\catcode
  `\&12\catcode `\#12\catcode `\^12\catcode `\_12\catcode `\%12\relax}%
\providecommand \@@startlink[1]{}%
\providecommand \@@endlink[0]{}%
\providecommand \url  [0]{\begingroup\@sanitize@url \@url }%
\providecommand \@url [1]{\endgroup\@href {#1}{\urlprefix }}%
\providecommand \urlprefix  [0]{URL }%
\providecommand \Eprint [0]{\href }%
\providecommand \doibase [0]{http://dx.doi.org/}%
\providecommand \selectlanguage [0]{\@gobble}%
\providecommand \bibinfo  [0]{\@secondoftwo}%
\providecommand \bibfield  [0]{\@secondoftwo}%
\providecommand \translation [1]{[#1]}%
\providecommand \BibitemOpen [0]{}%
\providecommand \bibitemStop [0]{}%
\providecommand \bibitemNoStop [0]{.\EOS\space}%
\providecommand \EOS [0]{\spacefactor3000\relax}%
\providecommand \BibitemShut  [1]{\csname bibitem#1\endcsname}%
\let\auto@bib@innerbib\@empty
%</preamble>
\bibitem [{\citenamefont {Horodecki}\ \emph {et~al.}(2009)\citenamefont
  {Horodecki}, \citenamefont {Horodecki}, \citenamefont {Horodecki},\ and\
  \citenamefont {Horodecki}}]{H09}%
  \BibitemOpen
  \bibfield  {author} {\bibinfo {author} {\bibfnamefont {R.}~\bibnamefont
  {Horodecki}}, \bibinfo {author} {\bibfnamefont {P.}~\bibnamefont
  {Horodecki}}, \bibinfo {author} {\bibfnamefont {M.}~\bibnamefont
  {Horodecki}}, \ and\ \bibinfo {author} {\bibfnamefont {K.}~\bibnamefont
  {Horodecki}},\ }\bibfield  {title} {\enquote {\bibinfo {title} {Quantum
  entanglement},}\ }\href {\doibase 10.1103/RevModPhys.81.865} {\bibfield
  {journal} {\bibinfo  {journal} {Rev. Mod. Phys.}\ }\textbf {\bibinfo {volume}
  {81}},\ \bibinfo {pages} {865} (\bibinfo {year} {2009})}\BibitemShut
  {NoStop}%
\bibitem [{\citenamefont {G\"uhne}\ and\ \citenamefont {Toth}(2009)}]{GT09}%
  \BibitemOpen
  \bibfield  {author} {\bibinfo {author} {\bibfnamefont {O.}~\bibnamefont
  {G\"uhne}}\ and\ \bibinfo {author} {\bibfnamefont {G.}~\bibnamefont {Toth}},\
  }\bibfield  {title} {\enquote {\bibinfo {title} {Entanglement detection},}\
  }\href {\doibase 10.1016/j.physrep.2009.02.004} {\bibfield  {journal}
  {\bibinfo  {journal} {Phys. Rep.}\ }\textbf {\bibinfo {volume} {474}},\
  \bibinfo {pages} {1} (\bibinfo {year} {2009})}\BibitemShut {NoStop}%
\bibitem [{\citenamefont {Eltschka}\ and\ \citenamefont
  {Siewert}(2014)}]{ES14}%
  \BibitemOpen
  \bibfield  {author} {\bibinfo {author} {\bibfnamefont {C.}~\bibnamefont
  {Eltschka}}\ and\ \bibinfo {author} {\bibfnamefont {J.}~\bibnamefont
  {Siewert}},\ }\bibfield  {title} {\enquote {\bibinfo {title} {Quantifying
  entanglement resources},}\ }\href {\doibase 10.1088/1751-8113/47/42/424005}
  {\bibfield  {journal} {\bibinfo  {journal} {J. Phys. A}\ }\textbf {\bibinfo
  {volume} {47}},\ \bibinfo {pages} {424005} (\bibinfo {year}
  {2014})}\BibitemShut {NoStop}%
\bibitem [{\citenamefont {Rosset}\ \emph {et~al.}(2012)\citenamefont {Rosset},
  \citenamefont {Ferretti-Sch\"{o}bitz}, \citenamefont {Bancal}, \citenamefont
  {Gisin},\ and\ \citenamefont {Liang}}]{RFD12}%
  \BibitemOpen
  \bibfield  {author} {\bibinfo {author} {\bibfnamefont {D.}~\bibnamefont
  {Rosset}}, \bibinfo {author} {\bibfnamefont {R.}~\bibnamefont
  {Ferretti-Sch\"{o}bitz}}, \bibinfo {author} {\bibfnamefont {J.-D.}\
  \bibnamefont {Bancal}}, \bibinfo {author} {\bibfnamefont {N.}~\bibnamefont
  {Gisin}}, \ and\ \bibinfo {author} {\bibfnamefont {Y.-C.}\ \bibnamefont
  {Liang}},\ }\bibfield  {title} {\enquote {\bibinfo {title} {Imperfect
  measurement settings: Implications for quantum state tomography and
  entanglement witnesses},}\ }\href {\doibase 10.1103/PhysRevA.86.062325}
  {\bibfield  {journal} {\bibinfo  {journal} {Phys. Rev. A}\ }\textbf {\bibinfo
  {volume} {86}},\ \bibinfo {pages} {062325} (\bibinfo {year}
  {2012})}\BibitemShut {NoStop}%
\bibitem [{\citenamefont {Ac\'{\i}n}\ \emph {et~al.}(2007)\citenamefont
  {Ac\'{\i}n}, \citenamefont {Brunner}, \citenamefont {Gisin}, \citenamefont
  {Massar}, \citenamefont {Pironio},\ and\ \citenamefont {Scarani}}]{Acin07}%
  \BibitemOpen
  \bibfield  {author} {\bibinfo {author} {\bibfnamefont {A.}~\bibnamefont
  {Ac\'{\i}n}}, \bibinfo {author} {\bibfnamefont {N.}~\bibnamefont {Brunner}},
  \bibinfo {author} {\bibfnamefont {N.}~\bibnamefont {Gisin}}, \bibinfo
  {author} {\bibfnamefont {S.}~\bibnamefont {Massar}}, \bibinfo {author}
  {\bibfnamefont {S.}~\bibnamefont {Pironio}}, \ and\ \bibinfo {author}
  {\bibfnamefont {V.}~\bibnamefont {Scarani}},\ }\bibfield  {title} {\enquote
  {\bibinfo {title} {Device-independent security of quantum cryptography
  against collective attacks},}\ }\href {\doibase
  10.1103/PhysRevLett.98.230501} {\bibfield  {journal} {\bibinfo  {journal}
  {Phys. Rev. Lett.}\ }\textbf {\bibinfo {volume} {98}},\ \bibinfo {pages}
  {230501} (\bibinfo {year} {2007})}\BibitemShut {NoStop}%
\bibitem [{\citenamefont {Colbeck}(2006)}]{Colbeck}%
  \BibitemOpen
  \bibfield  {author} {\bibinfo {author} {\bibfnamefont {R.}~\bibnamefont
  {Colbeck}},\ }\emph {\bibinfo {title} {Quantum And Relativistic Protocols For
  Secure Multi-Party Computation}},\ \href@noop {} {Ph.D. thesis},\ \bibinfo
  {school} {University of Cambridge} (\bibinfo {year} {2006}),\ \Eprint
  {http://arxiv.org/abs/0911.3814} {arXiv:0911.3814 [quant-ph]} \BibitemShut
  {NoStop}%
\bibitem [{\citenamefont {Pironio}\ \emph {et~al.}(2010)\citenamefont
  {Pironio}, \citenamefont {Ac\'{i}n}, \citenamefont {Massar}, \citenamefont
  {de~la Giroday}, \citenamefont {Matsukevich}, \citenamefont {Maunz},
  \citenamefont {Olmschenk}, \citenamefont {Hayes}, \citenamefont {Luo},
  \citenamefont {Manning},\ and\ \citenamefont {Monroe}}]{Pironio10}%
  \BibitemOpen
  \bibfield  {author} {\bibinfo {author} {\bibfnamefont {S.}~\bibnamefont
  {Pironio}}, \bibinfo {author} {\bibfnamefont {A.}~\bibnamefont {Ac\'{i}n}},
  \bibinfo {author} {\bibfnamefont {S.}~\bibnamefont {Massar}}, \bibinfo
  {author} {\bibfnamefont {A.~Boyer}\ \bibnamefont {de~la Giroday}}, \bibinfo
  {author} {\bibfnamefont {D.~N.}\ \bibnamefont {Matsukevich}}, \bibinfo
  {author} {\bibfnamefont {P.}~\bibnamefont {Maunz}}, \bibinfo {author}
  {\bibfnamefont {S.}~\bibnamefont {Olmschenk}}, \bibinfo {author}
  {\bibfnamefont {D.}~\bibnamefont {Hayes}}, \bibinfo {author} {\bibfnamefont
  {L.}~\bibnamefont {Luo}}, \bibinfo {author} {\bibfnamefont {T.~A.}\
  \bibnamefont {Manning}}, \ and\ \bibinfo {author} {\bibfnamefont
  {C.}~\bibnamefont {Monroe}},\ }\bibfield  {title} {\enquote {\bibinfo {title}
  {Random numbers certified by {B}ell's theorem},}\ }\href {\doibase
  10.1038/nature09008} {\bibfield  {journal} {\bibinfo  {journal} {Nature}\
  }\textbf {\bibinfo {volume} {464}},\ \bibinfo {pages} {1021} (\bibinfo {year}
  {2010})}\BibitemShut {NoStop}%
\bibitem [{\citenamefont {Collins}\ \emph
  {et~al.}(2002{\natexlab{a}})\citenamefont {Collins}, \citenamefont {Gisin},
  \citenamefont {Popescu}, \citenamefont {Roberts},\ and\ \citenamefont
  {Scarani}}]{Collins}%
  \BibitemOpen
  \bibfield  {author} {\bibinfo {author} {\bibfnamefont {D.}~\bibnamefont
  {Collins}}, \bibinfo {author} {\bibfnamefont {N.}~\bibnamefont {Gisin}},
  \bibinfo {author} {\bibfnamefont {S.}~\bibnamefont {Popescu}}, \bibinfo
  {author} {\bibfnamefont {D.}~\bibnamefont {Roberts}}, \ and\ \bibinfo
  {author} {\bibfnamefont {V.}~\bibnamefont {Scarani}},\ }\bibfield  {title}
  {\enquote {\bibinfo {title} {Bell-type inequalities to detect true
  $\mathit{n}$-body nonseparability},}\ }\href {\doibase
  10.1103/PhysRevLett.88.170405} {\bibfield  {journal} {\bibinfo  {journal}
  {Phys. Rev. Lett.}\ }\textbf {\bibinfo {volume} {88}},\ \bibinfo {pages}
  {170405} (\bibinfo {year} {2002}{\natexlab{a}})}\BibitemShut {NoStop}%
\bibitem [{\citenamefont {Seevinck}\ and\ \citenamefont
  {Svetlichny}(2002)}]{Seevinck}%
  \BibitemOpen
  \bibfield  {author} {\bibinfo {author} {\bibfnamefont {M.}~\bibnamefont
  {Seevinck}}\ and\ \bibinfo {author} {\bibfnamefont {G.}~\bibnamefont
  {Svetlichny}},\ }\bibfield  {title} {\enquote {\bibinfo {title} {Bell-type
  inequalities for partial separability in $n$-particle systems and quantum
  mechanical violations},}\ }\href {\doibase 10.1103/PhysRevLett.89.060401}
  {\bibfield  {journal} {\bibinfo  {journal} {Phys. Rev. Lett.}\ }\textbf
  {\bibinfo {volume} {89}},\ \bibinfo {pages} {060401} (\bibinfo {year}
  {2002})}\BibitemShut {NoStop}%
\bibitem [{\citenamefont {Bancal}\ \emph {et~al.}(2011)\citenamefont {Bancal},
  \citenamefont {Gisin}, \citenamefont {Liang},\ and\ \citenamefont
  {Pironio}}]{BG11}%
  \BibitemOpen
  \bibfield  {author} {\bibinfo {author} {\bibfnamefont {J.-D.}\ \bibnamefont
  {Bancal}}, \bibinfo {author} {\bibfnamefont {N.}~\bibnamefont {Gisin}},
  \bibinfo {author} {\bibfnamefont {Y.-C.}\ \bibnamefont {Liang}}, \ and\
  \bibinfo {author} {\bibfnamefont {S.}~\bibnamefont {Pironio}},\ }\bibfield
  {title} {\enquote {\bibinfo {title} {Device-independent witnesses of genuine
  multipartite entanglement},}\ }\href {\doibase
  10.1103/PhysRevLett.106.250404} {\bibfield  {journal} {\bibinfo  {journal}
  {Phys. Rev. Lett.}\ }\textbf {\bibinfo {volume} {106}},\ \bibinfo {pages}
  {250404} (\bibinfo {year} {2011})}\BibitemShut {NoStop}%
\bibitem [{\citenamefont {P\'{a}l}\ and\ \citenamefont
  {V\'{e}rtesi}(2011)}]{PV11}%
  \BibitemOpen
  \bibfield  {author} {\bibinfo {author} {\bibfnamefont {K.~F.}\ \bibnamefont
  {P\'{a}l}}\ and\ \bibinfo {author} {\bibfnamefont {T.}~\bibnamefont
  {V\'{e}rtesi}},\ }\bibfield  {title} {\enquote {\bibinfo {title}
  {Multisetting {B}ell-type inequalities for detecting genuine multipartite
  entanglement},}\ }\href {\doibase 10.1103/PhysRevA.83.062123} {\bibfield
  {journal} {\bibinfo  {journal} {Phys. Rev. A}\ }\textbf {\bibinfo {volume}
  {83}},\ \bibinfo {pages} {062123} (\bibinfo {year} {2011})}\BibitemShut
  {NoStop}%
\bibitem [{\citenamefont {Bancal}\ \emph {et~al.}(2012)\citenamefont {Bancal},
  \citenamefont {Branciard}, \citenamefont {Brunner}, \citenamefont {Gisin},\
  and\ \citenamefont {Liang}}]{BB12}%
  \BibitemOpen
  \bibfield  {author} {\bibinfo {author} {\bibfnamefont {J.-D.}\ \bibnamefont
  {Bancal}}, \bibinfo {author} {\bibfnamefont {C.}~\bibnamefont {Branciard}},
  \bibinfo {author} {\bibfnamefont {N.}~\bibnamefont {Brunner}}, \bibinfo
  {author} {\bibfnamefont {N.}~\bibnamefont {Gisin}}, \ and\ \bibinfo {author}
  {\bibfnamefont {Y.-C.}\ \bibnamefont {Liang}},\ }\bibfield  {title} {\enquote
  {\bibinfo {title} {A framework for the study of symmetric full-correlation
  {B}ell-like inequalities},}\ }\href {\doibase 10.1088/1751-8113/45/12/125301}
  {\bibfield  {journal} {\bibinfo  {journal} {J. Phys. A}\ }\textbf {\bibinfo
  {volume} {45}},\ \bibinfo {pages} {125301} (\bibinfo {year}
  {2012})}\BibitemShut {NoStop}%
\bibitem [{\citenamefont {Moroder}\ \emph {et~al.}(2013)\citenamefont
  {Moroder}, \citenamefont {Bancal}, \citenamefont {Liang}, \citenamefont
  {Hofmann},\ and\ \citenamefont {G\"uhne}}]{MB13}%
  \BibitemOpen
  \bibfield  {author} {\bibinfo {author} {\bibfnamefont {T.}~\bibnamefont
  {Moroder}}, \bibinfo {author} {\bibfnamefont {J.-D.}\ \bibnamefont {Bancal}},
  \bibinfo {author} {\bibfnamefont {Y.-C.}\ \bibnamefont {Liang}}, \bibinfo
  {author} {\bibfnamefont {M.}~\bibnamefont {Hofmann}}, \ and\ \bibinfo
  {author} {\bibfnamefont {O.}~\bibnamefont {G\"uhne}},\ }\bibfield  {title}
  {\enquote {\bibinfo {title} {Device-independent entanglement quantification
  and related applications},}\ }\href {\doibase 10.1103/PhysRevLett.111.030501}
  {\bibfield  {journal} {\bibinfo  {journal} {Phys. Rev. Lett.}\ }\textbf
  {\bibinfo {volume} {111}},\ \bibinfo {pages} {030501} (\bibinfo {year}
  {2013})}\BibitemShut {NoStop}%
\bibitem [{\citenamefont {Murta}\ \emph {et~al.}(2016)\citenamefont {Murta},
  \citenamefont {Ramanathan}, \citenamefont {M\'{o}ller},\ and\ \citenamefont
  {Cunha}}]{MR16}%
  \BibitemOpen
  \bibfield  {author} {\bibinfo {author} {\bibfnamefont {G.}~\bibnamefont
  {Murta}}, \bibinfo {author} {\bibfnamefont {R.}~\bibnamefont {Ramanathan}},
  \bibinfo {author} {\bibfnamefont {N.}~\bibnamefont {M\'{o}ller}}, \ and\
  \bibinfo {author} {\bibfnamefont {M.~Terra}\ \bibnamefont {Cunha}},\
  }\bibfield  {title} {\enquote {\bibinfo {title} {Quantum bounds on
  multiplayer linear games and device-independent witness of genuine tripartite
  entanglement},}\ }\href {\doibase 10.1103/PhysRevA.93.022305} {\bibfield
  {journal} {\bibinfo  {journal} {Phys. Rev. A}\ }\textbf {\bibinfo {volume}
  {93}},\ \bibinfo {pages} {022305} (\bibinfo {year} {2016})}\BibitemShut
  {NoStop}%
\bibitem [{\citenamefont {T\'oth}\ and\ \citenamefont
  {Ac\'{\i}n}(2006)}]{Toth}%
  \BibitemOpen
  \bibfield  {author} {\bibinfo {author} {\bibfnamefont {G.}~\bibnamefont
  {T\'oth}}\ and\ \bibinfo {author} {\bibfnamefont {A.}~\bibnamefont
  {Ac\'{\i}n}},\ }\bibfield  {title} {\enquote {\bibinfo {title} {Genuine
  tripartite entangled states with a local hidden-variable model},}\ }\href
  {\doibase 10.1103/PhysRevA.74.030306} {\bibfield  {journal} {\bibinfo
  {journal} {Phys. Rev. A}\ }\textbf {\bibinfo {volume} {74}},\ \bibinfo
  {pages} {030306} (\bibinfo {year} {2006})}\BibitemShut {NoStop}%
\bibitem [{\citenamefont {Augusiak}\ \emph {et~al.}(2015)\citenamefont
  {Augusiak}, \citenamefont {Demianowicz}, \citenamefont {Tura},\ and\
  \citenamefont {Ac\'{\i}n}}]{Remik}%
  \BibitemOpen
  \bibfield  {author} {\bibinfo {author} {\bibfnamefont {R.}~\bibnamefont
  {Augusiak}}, \bibinfo {author} {\bibfnamefont {M.}~\bibnamefont
  {Demianowicz}}, \bibinfo {author} {\bibfnamefont {J.}~\bibnamefont {Tura}}, \
  and\ \bibinfo {author} {\bibfnamefont {A.}~\bibnamefont {Ac\'{\i}n}},\
  }\bibfield  {title} {\enquote {\bibinfo {title} {Entanglement and nonlocality
  are inequivalent for any number of parties},}\ }\href {\doibase
  10.1103/PhysRevLett.115.030404} {\bibfield  {journal} {\bibinfo  {journal}
  {Phys. Rev. Lett.}\ }\textbf {\bibinfo {volume} {115}},\ \bibinfo {pages}
  {030404} (\bibinfo {year} {2015})}\BibitemShut {NoStop}%
\bibitem [{\citenamefont {Bowles}\ \emph {et~al.}(2016)\citenamefont {Bowles},
  \citenamefont {Francfort}, \citenamefont {Fillettaz}, \citenamefont
  {Hirsch},\ and\ \citenamefont {Brunner}}]{Bowles16}%
  \BibitemOpen
  \bibfield  {author} {\bibinfo {author} {\bibfnamefont {J.}~\bibnamefont
  {Bowles}}, \bibinfo {author} {\bibfnamefont {J.}~\bibnamefont {Francfort}},
  \bibinfo {author} {\bibfnamefont {M.}~\bibnamefont {Fillettaz}}, \bibinfo
  {author} {\bibfnamefont {F.}~\bibnamefont {Hirsch}}, \ and\ \bibinfo {author}
  {\bibfnamefont {N.}~\bibnamefont {Brunner}},\ }\bibfield  {title} {\enquote
  {\bibinfo {title} {Genuinely multipartite entangled quantum states with fully
  local hidden variable models and hidden multipartite nonlocality},}\ }\href
  {\doibase 10.1103/PhysRevLett.116.130401} {\bibfield  {journal} {\bibinfo
  {journal} {Phys. Rev. Lett.}\ }\textbf {\bibinfo {volume} {116}},\ \bibinfo
  {pages} {130401} (\bibinfo {year} {2016})}\BibitemShut {NoStop}%
\bibitem [{Note1()}]{Note1}%
  \BibitemOpen
  \bibinfo {note} {Note that nonlocality can nevertheless be activated in some
  more sophisticated Bell scenarios involving, e.g., sequential
  measurements~\cite {Popescu} or processing of multiple copies~\cite
  {Palazuelos}. Furthermore, with the help of additional sources of perfect
  singlets, any entangled bipartite state can be certified~\cite
  {BS18}.}\BibitemShut {Stop}%
\bibitem [{\citenamefont {Werner}(1989)}]{W89}%
  \BibitemOpen
  \bibfield  {author} {\bibinfo {author} {\bibfnamefont {R.~F.}\ \bibnamefont
  {Werner}},\ }\bibfield  {title} {\enquote {\bibinfo {title} {Quantum states
  with {E}instein-{P}odolsky-{R}osen correlations admitting a hidden-variable
  model},}\ }\href {\doibase 10.1103/PhysRevA.40.4277} {\bibfield  {journal}
  {\bibinfo  {journal} {Phys. Rev. A}\ }\textbf {\bibinfo {volume} {40}},\
  \bibinfo {pages} {4277} (\bibinfo {year} {1989})}\BibitemShut {NoStop}%
\bibitem [{\citenamefont {Augusiak}\ \emph {et~al.}(2014)\citenamefont
  {Augusiak}, \citenamefont {Demianowicz},\ and\ \citenamefont
  {Ac\'{i}n}}]{Augusiak_review}%
  \BibitemOpen
  \bibfield  {author} {\bibinfo {author} {\bibfnamefont {R.}~\bibnamefont
  {Augusiak}}, \bibinfo {author} {\bibfnamefont {M.}~\bibnamefont
  {Demianowicz}}, \ and\ \bibinfo {author} {\bibfnamefont {A.}~\bibnamefont
  {Ac\'{i}n}},\ }\bibfield  {title} {\enquote {\bibinfo {title} {Local hidden
  variable models for entangled quantum states},}\ }\href {\doibase
  10.1088/1751-8113/47/42/424002} {\bibfield  {journal} {\bibinfo  {journal}
  {J. Phys. A}\ }\textbf {\bibinfo {volume} {42}},\ \bibinfo {pages} {424002}
  (\bibinfo {year} {2014})}\BibitemShut {NoStop}%
\bibitem [{\citenamefont {Rabelo}\ \emph {et~al.}(2011)\citenamefont {Rabelo},
  \citenamefont {Ho}, \citenamefont {Cavalcanti}, \citenamefont {Brunner},\
  and\ \citenamefont {Scarani}}]{RH11}%
  \BibitemOpen
  \bibfield  {author} {\bibinfo {author} {\bibfnamefont {R.}~\bibnamefont
  {Rabelo}}, \bibinfo {author} {\bibfnamefont {M.}~\bibnamefont {Ho}}, \bibinfo
  {author} {\bibfnamefont {D.}~\bibnamefont {Cavalcanti}}, \bibinfo {author}
  {\bibfnamefont {N.}~\bibnamefont {Brunner}}, \ and\ \bibinfo {author}
  {\bibfnamefont {V.}~\bibnamefont {Scarani}},\ }\bibfield  {title} {\enquote
  {\bibinfo {title} {Device-independent certification of entangled
  measurements},}\ }\href {\doibase 10.1103/PhysRevLett.107.050502} {\bibfield
  {journal} {\bibinfo  {journal} {Phys. Rev. Lett.}\ }\textbf {\bibinfo
  {volume} {107}},\ \bibinfo {pages} {050502} (\bibinfo {year}
  {2011})}\BibitemShut {NoStop}%
\bibitem [{\citenamefont {He}\ and\ \citenamefont {Reid}(2013)}]{Reid}%
  \BibitemOpen
  \bibfield  {author} {\bibinfo {author} {\bibfnamefont {Q.~Y.}\ \bibnamefont
  {He}}\ and\ \bibinfo {author} {\bibfnamefont {M.~D.}\ \bibnamefont {Reid}},\
  }\bibfield  {title} {\enquote {\bibinfo {title} {Genuine multipartite
  {E}instein-{P}odolsky-{R}osen steering},}\ }\href {\doibase
  10.1103/PhysRevLett.111.250403} {\bibfield  {journal} {\bibinfo  {journal}
  {Phys. Rev. Lett.}\ }\textbf {\bibinfo {volume} {111}},\ \bibinfo {pages}
  {250403} (\bibinfo {year} {2013})}\BibitemShut {NoStop}%
\bibitem [{\citenamefont {Cavalcanti}\ \emph {et~al.}(2015)\citenamefont
  {Cavalcanti}, \citenamefont {Skrzypczyk}, \citenamefont {Aguilar},
  \citenamefont {Nery}, \citenamefont {Souto~Ribeiro},\ and\ \citenamefont
  {Walborn}}]{Paul}%
  \BibitemOpen
  \bibfield  {author} {\bibinfo {author} {\bibfnamefont {D.}~\bibnamefont
  {Cavalcanti}}, \bibinfo {author} {\bibfnamefont {P.}~\bibnamefont
  {Skrzypczyk}}, \bibinfo {author} {\bibfnamefont {G.~H.}\ \bibnamefont
  {Aguilar}}, \bibinfo {author} {\bibfnamefont {R.~V.}\ \bibnamefont {Nery}},
  \bibinfo {author} {\bibfnamefont {P.~H.}\ \bibnamefont {Souto~Ribeiro}}, \
  and\ \bibinfo {author} {\bibfnamefont {S.~P.}\ \bibnamefont {Walborn}},\
  }\bibfield  {title} {\enquote {\bibinfo {title} {Detection of entanglement in
  asymmetric quantum networks and multipartite quantum steering},}\ }\href
  {\doibase 10.1038/ncomms8941} {\bibfield  {journal} {\bibinfo  {journal}
  {Nat. Commun.}\ }\textbf {\bibinfo {volume} {6}},\ \bibinfo {pages} {7941}
  (\bibinfo {year} {2015})}\BibitemShut {NoStop}%
\bibitem [{\citenamefont {McCutcheon}\ \emph {et~al.}(2016)\citenamefont
  {McCutcheon}, \citenamefont {Pappa}, \citenamefont {Bell}, \citenamefont
  {McMillan}, \citenamefont {Chailloux}, \citenamefont {Lawson}, \citenamefont
  {Mafu}, \citenamefont {Markham}, \citenamefont {Diamanti}, \citenamefont
  {Kerenidis}, \citenamefont {Rarity},\ and\ \citenamefont {Tame}}]{Diamanti}%
  \BibitemOpen
  \bibfield  {author} {\bibinfo {author} {\bibfnamefont {W.}~\bibnamefont
  {McCutcheon}}, \bibinfo {author} {\bibfnamefont {A.}~\bibnamefont {Pappa}},
  \bibinfo {author} {\bibfnamefont {B.~A.}\ \bibnamefont {Bell}}, \bibinfo
  {author} {\bibfnamefont {A.}~\bibnamefont {McMillan}}, \bibinfo {author}
  {\bibfnamefont {A.}~\bibnamefont {Chailloux}}, \bibinfo {author}
  {\bibfnamefont {T.}~\bibnamefont {Lawson}}, \bibinfo {author} {\bibfnamefont
  {M.}~\bibnamefont {Mafu}}, \bibinfo {author} {\bibfnamefont {D.}~\bibnamefont
  {Markham}}, \bibinfo {author} {\bibfnamefont {E.}~\bibnamefont {Diamanti}},
  \bibinfo {author} {\bibfnamefont {I.}~\bibnamefont {Kerenidis}}, \bibinfo
  {author} {\bibfnamefont {J.~G.}\ \bibnamefont {Rarity}}, \ and\ \bibinfo
  {author} {\bibfnamefont {M.~S.}\ \bibnamefont {Tame}},\ }\bibfield  {title}
  {\enquote {\bibinfo {title} {Experimental verification of multipartite
  entanglement in quantum networks},}\ }\href {\doibase 10.1038/ncomms13251}
  {\bibfield  {journal} {\bibinfo  {journal} {Nat. Commun.}\ }\textbf {\bibinfo
  {volume} {7}},\ \bibinfo {pages} {13251} (\bibinfo {year}
  {2016})}\BibitemShut {NoStop}%
\bibitem [{\citenamefont {Buscemi}(2012)}]{B12}%
  \BibitemOpen
  \bibfield  {author} {\bibinfo {author} {\bibfnamefont {F.}~\bibnamefont
  {Buscemi}},\ }\bibfield  {title} {\enquote {\bibinfo {title} {All entangled
  quantum states are nonlocal},}\ }\href {\doibase
  10.1103/PhysRevLett.108.200401} {\bibfield  {journal} {\bibinfo  {journal}
  {Phys. Rev. Lett.}\ }\textbf {\bibinfo {volume} {108}},\ \bibinfo {pages}
  {200401} (\bibinfo {year} {2012})}\BibitemShut {NoStop}%
\bibitem [{\citenamefont {Branciard}\ \emph {et~al.}(2013)\citenamefont
  {Branciard}, \citenamefont {Rosset}, \citenamefont {Liang},\ and\
  \citenamefont {Gisin}}]{Branciard}%
  \BibitemOpen
  \bibfield  {author} {\bibinfo {author} {\bibfnamefont {C.}~\bibnamefont
  {Branciard}}, \bibinfo {author} {\bibfnamefont {D.}~\bibnamefont {Rosset}},
  \bibinfo {author} {\bibfnamefont {Y-C.}\ \bibnamefont {Liang}}, \ and\
  \bibinfo {author} {\bibfnamefont {N.}~\bibnamefont {Gisin}},\ }\bibfield
  {title} {\enquote {\bibinfo {title} {Measurement-device-independent
  entanglement witnesses for all entangled quantum states},}\ }\href {\doibase
  10.1103/PhysRevLett.110.060405} {\bibfield  {journal} {\bibinfo  {journal}
  {Phys. Rev. Lett.}\ }\textbf {\bibinfo {volume} {110}},\ \bibinfo {pages}
  {060405} (\bibinfo {year} {2013})}\BibitemShut {NoStop}%
\bibitem [{\citenamefont {Paw\l{}owski}\ and\ \citenamefont
  {Brunner}(2011)}]{PB11}%
  \BibitemOpen
  \bibfield  {author} {\bibinfo {author} {\bibfnamefont {M.}~\bibnamefont
  {Paw\l{}owski}}\ and\ \bibinfo {author} {\bibfnamefont {N.}~\bibnamefont
  {Brunner}},\ }\bibfield  {title} {\enquote {\bibinfo {title}
  {Semi-device-independent security of one-way quantum key distribution},}\
  }\href {\doibase 10.1103/PhysRevA.84.010302} {\bibfield  {journal} {\bibinfo
  {journal} {Phys. Rev. A}\ }\textbf {\bibinfo {volume} {84}},\ \bibinfo
  {pages} {010302} (\bibinfo {year} {2011})}\BibitemShut {NoStop}%
\bibitem [{\citenamefont {Li}\ \emph {et~al.}(2012)\citenamefont {Li},
  \citenamefont {Paw\l{}owski}, \citenamefont {Yin}, \citenamefont {Guo},\ and\
  \citenamefont {Han}}]{LP12}%
  \BibitemOpen
  \bibfield  {author} {\bibinfo {author} {\bibfnamefont {H.-W.}\ \bibnamefont
  {Li}}, \bibinfo {author} {\bibfnamefont {M.}~\bibnamefont {Paw\l{}owski}},
  \bibinfo {author} {\bibfnamefont {Z.-Q.}\ \bibnamefont {Yin}}, \bibinfo
  {author} {\bibfnamefont {G.-C.}\ \bibnamefont {Guo}}, \ and\ \bibinfo
  {author} {\bibfnamefont {Z.-F.}\ \bibnamefont {Han}},\ }\bibfield  {title}
  {\enquote {\bibinfo {title} {Semi-device-independent randomness certification
  using {$n\to1$} quantum random access codes},}\ }\href {\doibase
  10.1103/PhysRevA.85.052308} {\bibfield  {journal} {\bibinfo  {journal} {Phys.
  Rev. A}\ }\textbf {\bibinfo {volume} {85}},\ \bibinfo {pages} {052308}
  (\bibinfo {year} {2012})}\BibitemShut {NoStop}%
\bibitem [{\citenamefont {Tavakoli}\ \emph {et~al.}(2015)\citenamefont
  {Tavakoli}, \citenamefont {Hameedi}, \citenamefont {Marques},\ and\
  \citenamefont {Bourennane}}]{TH15}%
  \BibitemOpen
  \bibfield  {author} {\bibinfo {author} {\bibfnamefont {A.}~\bibnamefont
  {Tavakoli}}, \bibinfo {author} {\bibfnamefont {A.}~\bibnamefont {Hameedi}},
  \bibinfo {author} {\bibfnamefont {B.}~\bibnamefont {Marques}}, \ and\
  \bibinfo {author} {\bibfnamefont {M.}~\bibnamefont {Bourennane}},\ }\bibfield
   {title} {\enquote {\bibinfo {title} {Quantum random access codes using
  single $d$-level systems},}\ }\href {\doibase 10.1103/PhysRevLett.114.170502}
  {\bibfield  {journal} {\bibinfo  {journal} {Phys. Rev. Lett.}\ }\textbf
  {\bibinfo {volume} {114}},\ \bibinfo {pages} {170502} (\bibinfo {year}
  {2015})}\BibitemShut {NoStop}%
\bibitem [{\citenamefont {Liang}\ \emph {et~al.}(2011)\citenamefont {Liang},
  \citenamefont {V\'{e}rtesi},\ and\ \citenamefont {Brunner}}]{LV11}%
  \BibitemOpen
  \bibfield  {author} {\bibinfo {author} {\bibfnamefont {Y.-C.}\ \bibnamefont
  {Liang}}, \bibinfo {author} {\bibfnamefont {T.}~\bibnamefont {V\'{e}rtesi}},
  \ and\ \bibinfo {author} {\bibfnamefont {N.}~\bibnamefont {Brunner}},\
  }\bibfield  {title} {\enquote {\bibinfo {title} {Semi-device-independent
  bounds on entanglement},}\ }\href {\doibase 10.1103/PhysRevA.83.022108}
  {\bibfield  {journal} {\bibinfo  {journal} {Phys. Rev. A}\ }\textbf {\bibinfo
  {volume} {83}},\ \bibinfo {pages} {022108} (\bibinfo {year}
  {2011})}\BibitemShut {NoStop}%
\bibitem [{\citenamefont {Goh}\ \emph {et~al.}(2016)\citenamefont {Goh},
  \citenamefont {Bancal},\ and\ \citenamefont {Scarani}}]{Koon}%
  \BibitemOpen
  \bibfield  {author} {\bibinfo {author} {\bibfnamefont {K.~T.}\ \bibnamefont
  {Goh}}, \bibinfo {author} {\bibfnamefont {J.-D.}\ \bibnamefont {Bancal}}, \
  and\ \bibinfo {author} {\bibfnamefont {V.}~\bibnamefont {Scarani}},\
  }\bibfield  {title} {\enquote {\bibinfo {title}
  {Measurement-device-independent quantification of entanglement for given
  {H}ilbert space dimension},}\ }\href {\doibase 10.1088/1367-2630/18/4/045022}
  {\bibfield  {journal} {\bibinfo  {journal} {New J. Phys.}\ }\textbf {\bibinfo
  {volume} {18}},\ \bibinfo {pages} {045022} (\bibinfo {year}
  {2016})}\BibitemShut {NoStop}%
\bibitem [{\citenamefont {V\'{e}rtesi}\ and\ \citenamefont
  {Navascu\'{e}s}(2011)}]{VN11}%
  \BibitemOpen
  \bibfield  {author} {\bibinfo {author} {\bibfnamefont {T.}~\bibnamefont
  {V\'{e}rtesi}}\ and\ \bibinfo {author} {\bibfnamefont {M.}~\bibnamefont
  {Navascu\'{e}s}},\ }\bibfield  {title} {\enquote {\bibinfo {title}
  {Certifying entangled measurements in known {H}ilbert spaces},}\ }\href
  {\doibase 10.1103/PhysRevA.83.062112} {\bibfield  {journal} {\bibinfo
  {journal} {Phys. Rev. A}\ }\textbf {\bibinfo {volume} {83}},\ \bibinfo
  {pages} {062112} (\bibinfo {year} {2011})}\BibitemShut {NoStop}%
\bibitem [{\citenamefont {Bennet}\ \emph {et~al.}(2014)\citenamefont {Bennet},
  \citenamefont {V\'{e}rtesi}, \citenamefont {Saunders}, \citenamefont
  {Brunner},\ and\ \citenamefont {Pryde}}]{BV14}%
  \BibitemOpen
  \bibfield  {author} {\bibinfo {author} {\bibfnamefont {A.}~\bibnamefont
  {Bennet}}, \bibinfo {author} {\bibfnamefont {T.}~\bibnamefont {V\'{e}rtesi}},
  \bibinfo {author} {\bibfnamefont {D.~J.}\ \bibnamefont {Saunders}}, \bibinfo
  {author} {\bibfnamefont {N.}~\bibnamefont {Brunner}}, \ and\ \bibinfo
  {author} {\bibfnamefont {G.~J.}\ \bibnamefont {Pryde}},\ }\bibfield  {title}
  {\enquote {\bibinfo {title} {Experimental semi-device-independent
  certification of entangled measurements},}\ }\href {\doibase
  10.1103/PhysRevLett.113.080405} {\bibfield  {journal} {\bibinfo  {journal}
  {Phys. Rev. Lett.}\ }\textbf {\bibinfo {volume} {113}},\ \bibinfo {pages}
  {080405} (\bibinfo {year} {2014})}\BibitemShut {NoStop}%
\bibitem [{\citenamefont {Bennett}\ \emph {et~al.}(1993)\citenamefont
  {Bennett}, \citenamefont {Brassard}, \citenamefont {Cr\'{e}peau},
  \citenamefont {Jozsa}, \citenamefont {Peres},\ and\ \citenamefont
  {Wootters}}]{BB93}%
  \BibitemOpen
  \bibfield  {author} {\bibinfo {author} {\bibfnamefont {C.~H.}\ \bibnamefont
  {Bennett}}, \bibinfo {author} {\bibfnamefont {G.}~\bibnamefont {Brassard}},
  \bibinfo {author} {\bibfnamefont {C.}~\bibnamefont {Cr\'{e}peau}}, \bibinfo
  {author} {\bibfnamefont {R.}~\bibnamefont {Jozsa}}, \bibinfo {author}
  {\bibfnamefont {A.}~\bibnamefont {Peres}}, \ and\ \bibinfo {author}
  {\bibfnamefont {W.~K.}\ \bibnamefont {Wootters}},\ }\bibfield  {title}
  {\enquote {\bibinfo {title} {Teleporting an unknown quantum state via dual
  classical and {E}instein-{P}odolsky-{R}osen channels},}\ }\href {\doibase
  10.1103/PhysRevLett.70.1895} {\bibfield  {journal} {\bibinfo  {journal}
  {Phys. Rev. Lett.}\ }\textbf {\bibinfo {volume} {70}},\ \bibinfo {pages}
  {1895} (\bibinfo {year} {1993})}\BibitemShut {NoStop}%
\bibitem [{\citenamefont {Horodecki}\ \emph {et~al.}(1999)\citenamefont
  {Horodecki}, \citenamefont {Horodecki},\ and\ \citenamefont
  {Horodecki}}]{HHH99}%
  \BibitemOpen
  \bibfield  {author} {\bibinfo {author} {\bibfnamefont {M.}~\bibnamefont
  {Horodecki}}, \bibinfo {author} {\bibfnamefont {P.}~\bibnamefont
  {Horodecki}}, \ and\ \bibinfo {author} {\bibfnamefont {R.}~\bibnamefont
  {Horodecki}},\ }\bibfield  {title} {\enquote {\bibinfo {title} {General
  teleportation channel, singlet fraction, and quasidistillation},}\ }\href
  {\doibase 10.1103/PhysRevA.60.1888} {\bibfield  {journal} {\bibinfo
  {journal} {Phys. Rev. A}\ }\textbf {\bibinfo {volume} {60}},\ \bibinfo
  {pages} {1888} (\bibinfo {year} {1999})}\BibitemShut {NoStop}%
\bibitem [{\citenamefont {Hirsch}\ \emph {et~al.}(2017)\citenamefont {Hirsch},
  \citenamefont {Quintino}, \citenamefont {V\'{e}rtesi}, \citenamefont
  {Navascu\'{e}s},\ and\ \citenamefont {Brunner}}]{HQ17}%
  \BibitemOpen
  \bibfield  {author} {\bibinfo {author} {\bibfnamefont {F.}~\bibnamefont
  {Hirsch}}, \bibinfo {author} {\bibfnamefont {M.~T.}\ \bibnamefont
  {Quintino}}, \bibinfo {author} {\bibfnamefont {T.}~\bibnamefont
  {V\'{e}rtesi}}, \bibinfo {author} {\bibfnamefont {M.}~\bibnamefont
  {Navascu\'{e}s}}, \ and\ \bibinfo {author} {\bibfnamefont {N.}~\bibnamefont
  {Brunner}},\ }\bibfield  {title} {\enquote {\bibinfo {title} {Better local
  hidden variable models for two-qubit {W}erner states and an upper bound on
  the {G}rothendieck constant {$K_G(3)$}},}\ }\href {\doibase
  10.22331/q-2017-04-25-3} {\bibfield  {journal} {\bibinfo  {journal}
  {Quantum}\ }\textbf {\bibinfo {volume} {1}},\ \bibinfo {pages} {3} (\bibinfo
  {year} {2017})}\BibitemShut {NoStop}%
\bibitem [{\citenamefont {Almeida}\ \emph {et~al.}(2007)\citenamefont
  {Almeida}, \citenamefont {Pironio}, \citenamefont {Barrett}, \citenamefont
  {T\'oth},\ and\ \citenamefont {Ac\'in}}]{Almeida}%
  \BibitemOpen
  \bibfield  {author} {\bibinfo {author} {\bibfnamefont {M.~L.}\ \bibnamefont
  {Almeida}}, \bibinfo {author} {\bibfnamefont {S.}~\bibnamefont {Pironio}},
  \bibinfo {author} {\bibfnamefont {J.}~\bibnamefont {Barrett}}, \bibinfo
  {author} {\bibfnamefont {G.}~\bibnamefont {T\'oth}}, \ and\ \bibinfo {author}
  {\bibfnamefont {A.}~\bibnamefont {Ac\'in}},\ }\bibfield  {title} {\enquote
  {\bibinfo {title} {Noise robustness of the nonlocality of entangled quantum
  states},}\ }\href {\doibase 10.1103/PhysRevLett.99.040403} {\bibfield
  {journal} {\bibinfo  {journal} {Phys. Rev. Lett.}\ }\textbf {\bibinfo
  {volume} {99}},\ \bibinfo {pages} {040403} (\bibinfo {year}
  {2007})}\BibitemShut {NoStop}%
\bibitem [{\citenamefont {Collins}\ \emph
  {et~al.}(2002{\natexlab{b}})\citenamefont {Collins}, \citenamefont {Gisin},
  \citenamefont {Linden}, \citenamefont {Massar},\ and\ \citenamefont
  {Popescu}}]{CGLMP02}%
  \BibitemOpen
  \bibfield  {author} {\bibinfo {author} {\bibfnamefont {D.}~\bibnamefont
  {Collins}}, \bibinfo {author} {\bibfnamefont {N.}~\bibnamefont {Gisin}},
  \bibinfo {author} {\bibfnamefont {N.}~\bibnamefont {Linden}}, \bibinfo
  {author} {\bibfnamefont {S.}~\bibnamefont {Massar}}, \ and\ \bibinfo {author}
  {\bibfnamefont {S.}~\bibnamefont {Popescu}},\ }\bibfield  {title} {\enquote
  {\bibinfo {title} {Bell inequalities for arbitrarily high-dimensional
  systems},}\ }\href {\doibase 10.1103/PhysRevLett.88.040404} {\bibfield
  {journal} {\bibinfo  {journal} {Phys. Rev. Lett.}\ }\textbf {\bibinfo
  {volume} {88}},\ \bibinfo {pages} {040404} (\bibinfo {year}
  {2002}{\natexlab{b}})}\BibitemShut {NoStop}%
\bibitem [{\citenamefont {Rafsanjani}\ \emph {et~al.}(2012)\citenamefont
  {Rafsanjani}, \citenamefont {Huber}, \citenamefont {Broadbent},\ and\
  \citenamefont {Eberly}}]{RH12}%
  \BibitemOpen
  \bibfield  {author} {\bibinfo {author} {\bibfnamefont {S.~M.~Hashemi}\
  \bibnamefont {Rafsanjani}}, \bibinfo {author} {\bibfnamefont
  {M.}~\bibnamefont {Huber}}, \bibinfo {author} {\bibfnamefont {C.~J.}\
  \bibnamefont {Broadbent}}, \ and\ \bibinfo {author} {\bibfnamefont {J.~H.}\
  \bibnamefont {Eberly}},\ }\bibfield  {title} {\enquote {\bibinfo {title}
  {Genuinely multipartite concurrence of {$N$}-qubit {$X$}-matrices},}\ }\href
  {\doibase 10.1103/PhysRevA.86.062303} {\bibfield  {journal} {\bibinfo
  {journal} {Phys. Rev. A}\ }\textbf {\bibinfo {volume} {86}},\ \bibinfo
  {pages} {062303} (\bibinfo {year} {2012})}\BibitemShut {NoStop}%
\bibitem [{\citenamefont {Clivaz}\ \emph {et~al.}(2017)\citenamefont {Clivaz},
  \citenamefont {Huber}, \citenamefont {Lami},\ and\ \citenamefont
  {Murta}}]{CH17}%
  \BibitemOpen
  \bibfield  {author} {\bibinfo {author} {\bibfnamefont {F.}~\bibnamefont
  {Clivaz}}, \bibinfo {author} {\bibfnamefont {M.}~\bibnamefont {Huber}},
  \bibinfo {author} {\bibfnamefont {L.}~\bibnamefont {Lami}}, \ and\ \bibinfo
  {author} {\bibfnamefont {G.}~\bibnamefont {Murta}},\ }\bibfield  {title}
  {\enquote {\bibinfo {title} {Genuine-multipartite entanglement criteria based
  on positive maps},}\ }\href {\doibase 10.1063/1.4998433} {\bibfield
  {journal} {\bibinfo  {journal} {J. Math. Phys.}\ }\textbf {\bibinfo {volume}
  {58}},\ \bibinfo {pages} {082201} (\bibinfo {year} {2017})}\BibitemShut
  {NoStop}%
\bibitem [{\citenamefont {Xu}(2016)}]{Xu16}%
  \BibitemOpen
  \bibfield  {author} {\bibinfo {author} {\bibfnamefont {J.}~\bibnamefont
  {Xu}},\ }\bibfield  {title} {\enquote {\bibinfo {title} {Multipartite fully
  entangled fraction},}\ }\href {\doibase 10.1007/s10773-016-2921-2} {\bibfield
   {journal} {\bibinfo  {journal} {Int. J. Theor. Phys.}\ }\textbf {\bibinfo
  {volume} {55}},\ \bibinfo {pages} {2904} (\bibinfo {year}
  {2016})}\BibitemShut {NoStop}%
\bibitem [{\citenamefont {P.~Badziag}\ and\ \citenamefont
  {Horodecki}(2000)}]{badziag}%
  \BibitemOpen
  \bibfield  {author} {\bibinfo {author} {\bibfnamefont {P.~Horodecki}\
  \bibnamefont {P.~Badziag}, \bibfnamefont {M.~Horodecki}}\ and\ \bibinfo
  {author} {\bibfnamefont {R.}~\bibnamefont {Horodecki}},\ }\bibfield  {title}
  {\enquote {\bibinfo {title} {Local environment can enhance fidelity of
  quantum teleportation},}\ }\href {\doibase 10.1103/PhysRevA.62.012311}
  {\bibfield  {journal} {\bibinfo  {journal} {Phys. Rev. A}\ }\textbf {\bibinfo
  {volume} {62}},\ \bibinfo {pages} {012311} (\bibinfo {year}
  {2000})}\BibitemShut {NoStop}%
\bibitem [{\citenamefont {Jungnitsch}\ \emph {et~al.}(2011)\citenamefont
  {Jungnitsch}, \citenamefont {Moroder},\ and\ \citenamefont
  {G\"{u}hne}}]{JM11}%
  \BibitemOpen
  \bibfield  {author} {\bibinfo {author} {\bibfnamefont {B.}~\bibnamefont
  {Jungnitsch}}, \bibinfo {author} {\bibfnamefont {T.}~\bibnamefont {Moroder}},
  \ and\ \bibinfo {author} {\bibfnamefont {O.}~\bibnamefont {G\"{u}hne}},\
  }\bibfield  {title} {\enquote {\bibinfo {title} {Taming multiparticle
  entanglement},}\ }\href {\doibase 10.1103/PhysRevLett.106.190502} {\bibfield
  {journal} {\bibinfo  {journal} {Phys. Rev. Lett.}\ }\textbf {\bibinfo
  {volume} {106}},\ \bibinfo {pages} {190502} (\bibinfo {year}
  {2011})}\BibitemShut {NoStop}%
\bibitem [{\citenamefont {Tavakoli}\ \emph {et~al.}()\citenamefont {Tavakoli},
  \citenamefont {Kaniewski}, \citenamefont {V\'{e}rtesi}, \citenamefont
  {Rosset},\ and\ \citenamefont {Brunner}}]{TK18}%
  \BibitemOpen
  \bibfield  {author} {\bibinfo {author} {\bibfnamefont {A.}~\bibnamefont
  {Tavakoli}}, \bibinfo {author} {\bibfnamefont {J.}~\bibnamefont {Kaniewski}},
  \bibinfo {author} {\bibfnamefont {T.}~\bibnamefont {V\'{e}rtesi}}, \bibinfo
  {author} {\bibfnamefont {D.}~\bibnamefont {Rosset}}, \ and\ \bibinfo {author}
  {\bibfnamefont {N.}~\bibnamefont {Brunner}},\ }\bibfield  {title} {\enquote
  {\bibinfo {title} {Self-testing quantum states and measurements in the
  prepare-and-measure scenario},}\ }\href@noop {} {\ }\Eprint
  {http://arxiv.org/abs/1801.08520} {arXiv:1801.08520 [quant-ph]} \BibitemShut
  {NoStop}%
\bibitem [{\citenamefont {Popescu}(1995)}]{Popescu}%
  \BibitemOpen
  \bibfield  {author} {\bibinfo {author} {\bibfnamefont {S.}~\bibnamefont
  {Popescu}},\ }\bibfield  {title} {\enquote {\bibinfo {title} {Bell's
  inequalities and density matrices: Revealing ``hidden'' nonlocality},}\
  }\href {\doibase 10.1103/PhysRevLett.74.2619} {\bibfield  {journal} {\bibinfo
   {journal} {Phys. Rev. Lett.}\ }\textbf {\bibinfo {volume} {74}},\ \bibinfo
  {pages} {2619--2622} (\bibinfo {year} {1995})}\BibitemShut {NoStop}%
\bibitem [{\citenamefont {Palazuelos}(2012)}]{Palazuelos}%
  \BibitemOpen
  \bibfield  {author} {\bibinfo {author} {\bibfnamefont {C.}~\bibnamefont
  {Palazuelos}},\ }\bibfield  {title} {\enquote {\bibinfo {title}
  {Superactivation of quantum nonlocality},}\ }\href {\doibase
  10.1103/PhysRevLett.109.190401} {\bibfield  {journal} {\bibinfo  {journal}
  {Phys. Rev. Lett.}\ }\textbf {\bibinfo {volume} {109}},\ \bibinfo {pages}
  {190401} (\bibinfo {year} {2012})}\BibitemShut {NoStop}%
\bibitem [{\citenamefont {Bowles}\ \emph {et~al.}()\citenamefont {Bowles},
  \citenamefont {\v{S}upi\'{c}}, \citenamefont {Cavalcanti},\ and\
  \citenamefont {Ac\'{i}n}}]{BS18}%
  \BibitemOpen
  \bibfield  {author} {\bibinfo {author} {\bibfnamefont {J.}~\bibnamefont
  {Bowles}}, \bibinfo {author} {\bibfnamefont {I.}~\bibnamefont
  {\v{S}upi\'{c}}}, \bibinfo {author} {\bibfnamefont {D.}~\bibnamefont
  {Cavalcanti}}, \ and\ \bibinfo {author} {\bibfnamefont {A.}~\bibnamefont
  {Ac\'{i}n}},\ }\bibfield  {title} {\enquote {\bibinfo {title}
  {Device-independent entanglement certification of all entangled states},}\
  }\href@noop {} {\ }\Eprint {http://arxiv.org/abs/1801.10444}
  {arXiv:1801.10444 [quant-ph]} \BibitemShut {NoStop}%
\end{thebibliography}%

\appendix
\onecolumngrid
\newpage
\section{Proof of Result \ref{res1}}\label{A}

In this section, we prove Result~\ref{res1}. 
Specifically, we bound the maximal value of $\mathcal{A}_{n,d}$ that can be obtained by biseparable states, regardless of the choice of transformations and measurements. 
Since $\mathcal{A}_{n,d}$ depends linearly on $\rho$, no mixed biseparable state can be used to obtain a larger score than some pure biseparable state. 
Hence, we need only consider pure states of the form $\ket{\chi}_S\equiv\ket{\psi}_S\otimes \ket{\phi}_{\bar{S}}$, for any nontrivial bipartition $\{S,\bar{S}\}$ of the set of subsystems $\{1,\ldots,n\}$. 

Consider $\ket{\chi}_S$ for a particular bipartition $\{S,\bar{S}\}$. 
In order to give an upper bound on $\mathcal{A}_{n,d}$, we will relax some of the constraints in the scenario considered by the scheme and then evaluate the maximal average score in this less restrictive setting. 
In particular, we consider the scenario in which the parties $\{A_k\}_{k\in S}$ are permitted to communicate unbounded information to $B$, while the remaining parties $\{A_k\}_{k\in \bar{S}}$ are grouped into a single ``effective'' party $R$, and which receives all their inputs. 
The party $R$ is allowed to send $|\bar{S}|$ $d$-level quantum systems (equivalently, a system of dimension $d^{|\bar{S}|}$) to  $B$. 
This scenario is illustrated in Fig.~\ref{Fig2}.

\begin{figure}
	\centering
	\includegraphics[width=0.42\columnwidth]{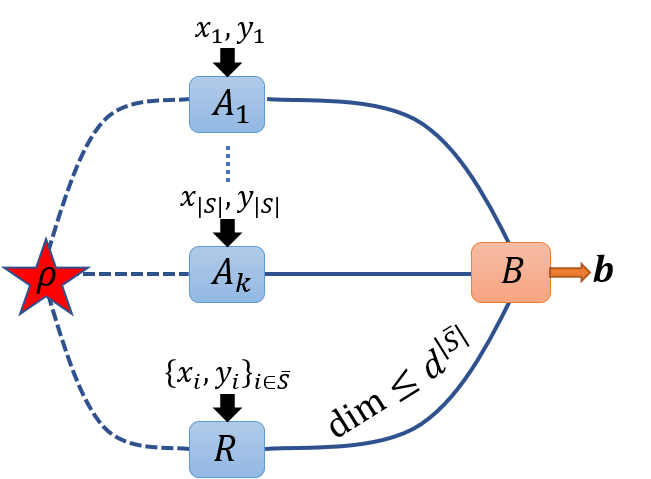}
	\caption{The modified scenario (cf.\ Fig.~\ref{fig}). Parties $\{A_k\}_{k\in S}$ are allowed unlimited communication to $B$, whereas the remaining parties $\{A_k\}_{k\in \bar{S}}$ are grouped into a single party, $R$, allowed $\bar{S}\log d$ bits of communication.
	The case shown is for $S=\{1,\dots,|S|\}$.
		}\label{Fig2}
\end{figure}

It is clear that any probability distribution $P(\mathbf{b}|\mathbf{x},\mathbf{y})$ obtainable in the original scenario on a state $\ket{\chi}_S$ (i.e., with some choice of transformations and measurements), is also obtainable in the relaxed scenario, but not vice versa. 
Since, by assumption, there is no entanglement over the bipartition $\{S,\bar{S}\}$, the parties $\{A_k\}_{k\in S}$ cannot do better than to simply send all their inputs to $B$. 
In order to win the game, $R$ therefore needs to communicate to $B$ the values of $\sum_{i\in \bar{S}} x_i$ and all $\{y_i\}_{i\in \bar{S}}$ for the conditions $C_1,\ldots, C_n$ to be satisfied. 
However, this amounts to $\left(|\bar{S}|+1\right)\log d$ bits of information, while $R$ can generally only send $|\bar{S}|\log d$ bits using a $d^{|\bar{S}|}$-level system, and must therefore employ a nontrivial optimal communication strategy. 
As we demonstrate more formally below, there is no quantum strategy that allows $B$ to know this information with a probability higher than $1/d$, and therefore cannot with the game with a probability better than this either. 
Consequently, the desired bound in Eq.~\eqref{state} follows. 

To this end, let us consider the following general setting for an information compression task between two parties, Alice and Bob.
Let Alice receive a uniformly distributed input $x\in\{1,\ldots,N\}$ which she must communicate to Bob. 
She encodes this input into a quantum state $\rho_x$ of dimension at most $d$, with $N>d$ (if $N\leq d$ then Alice can simply send $x$). 
This amounts to Alice compressing her input into a smaller message to send. 
This message is then sent to Bob, who must attempt to retrieve the value of $x$ from the state $\rho_x$ by performing a suitable measurement $\{M_b\}_{b=1}^{N}$. 
What is the average probability of success for Bob to correctly obtain $x$ following this measurement? 
A quantum strategy must obey the following bound:
\begin{align}\label{qbound}
p^{\text{success}}\equiv\frac{1}{N}\sum_{x=1}^{N}P^{\text{Q}}(b=x|x)= \frac{1}{N}\sum_{x=1}^{N}\Tr\left(\rho_xM_x\right) \leq \frac{1}{N}\sum_{x=1}^{N}\lambda_{\text{max}}\left(M_x\right) \leq\frac{1}{N}\sum_{x=1}^{N}\Tr\left(M_x\right)=\frac{1}{N}\Tr\left(\sum_{x=1}^{N}M_x\right)=\frac{d}{N},
\end{align}
where we have used the fact that the optimal state $\rho_x$ maximising $\Tr(\rho_x M_x)$ is the eigenvector of $M_x$ corresponding to its largest eigenvalue, and that $\lambda_{\text{max}}(M_x)\leq\Tr\left(M_x\right)$. In the last step we used that $\sum_{x}M_x=\id_d$ for any POVM $\{M_x\}_x$. 

Moreover, there is no quantum advantage over classical approaches to encode and decode this information. This is straightforwardly seen by noting that the quantum bound~\eqref{qbound} can be saturated with a classical strategy as follows. 
Let Alice send the classical message $m(x)=x$ whenever $x=1,\ldots,d$, and $m(x)=d$ if $x\in\{d+1,\ldots,N\}$. Bob always outputs the message he receives, i.e.\ $b=m(x)$. Whenever $x\in\{1,\ldots,d\}$, it is indeed true that $b=x$ and Bob correctly obtains $x$; when $x\in\{d+1,\ldots,N\}$, Bob never succeeds. The average success probability of this simple classical strategy reads
\begin{equation}
p^{\text{success}}\equiv \frac{1}{N}\sum_{x=1}^{N}P^{\text{C}}(b=x|x)=\frac{d}{N}.
\end{equation}
Hence, quantum theory provides no advantage over classical coding for the stated task. Note that the above bears resemblance to the Holevo bound, with the difference that we are probabilistically accessing the information.

We can now apply this result to the relaxed scenario under consideration in the proof of Result~\ref{res1}.
There, the effective party $R$ plays the role of Alice, and has to transmit to Bob which of the $d^{|\bar{S}|+1}$ possible inputs they received by encoding it in a $d^{|\bar{S}|}$-dimensional quantum system. 
From the above result, the average success probability of $B$ correctly guessing the inputs of $R$---and therefore winning the game---cannot be better than $1/d$, which is indeed the desired result.

\section{Proof of Result \ref{res2}}\label{B}

Here we present the proof of Eq.~\eqref{measurement}. We begin with a useful lemma, which is straightforward:
\begin{lemma}\label{lemma}
Let $\sigma_{AB}$ be a bipartite density matrix in $\mH_A\otimes\mH_B$ and $0\leq M\leq\id$ (resp.\ $N$) be a symmetric operator on $\mH_A$ (resp. $\mH_B$). Let $\sigma_{A}=\trr{B}{\sigma_{AB}}$.
Then the following inequality holds:
\begin{equation}
	\trr{AB}{\sigma_{AB}(M\otimes N)}\leq \trr{A}{\sigma_A M}\lambda_{\mathrm{max}}[N].
\end{equation}
\end{lemma}
\begin{proof}
Let $N=\sum_i n_i\ketbra{i}{i}$ be the spectral decomposition of $N$. Remark that $\sigma^{(i)}_A=\bracket{i}{\sigma_{AB}}{i}$ is an (unnormalised) positive semidefinite operator. We then have $\trr{AB}{\sigma_{AB}(M\otimes N)}=\sum_i n_i\trr{A}{\sigma^{(i)}_A M} \leq \lambda_{\mathrm{max}}[N]\trr{A}{\sum_i\sigma^{(i)}_A M}=\lambda_{\mathrm{max}}[N]\trr{A}{\sigma_A M}$.
\end{proof}

Equipped with this lemma, we can now prove the statement of Eq.~\eqref{measurement}. 
Party $B$ performs a measurement with $d^n$ different possible outcomes, described by measurement operators $\{M_\mathbf{b}\}_\mathbf{b}$ with $M_\mathbf{b}\geq 0$ and $\sum_\mathbf{b} M_\mathbf{b} = \id$.
Let $\text{SEP}$ be the set of strings $\mathbf{b}$ for which $M_\mathbf{b}$ is a separable measurement operator, with $\lvert\text{SEP}\rvert\ge k$;
for $\mathbf{b}\in\text{SEP}$ we thus have $M_\mathbf{b} = \sum_i \bigotimes_{k=1}^n M^{(k)}_{\mathbf{b},i}$.
We will initially assume for simplicity that these separable POVM elements have the simpler tensor product form $M_\mathbf{b}=\bigotimes_{k=1}^n M^{(k)}_\mathbf{b}$; we will see later that, because of the linearity of $\mathcal{A}_{n,d}$, the proof nonetheless holds for arbitrary separable operators.

The parties $A_1,\ldots,A_n$ may perform arbitrary transformations $\mathcal{T}^{(k)}_{x_ky_k}$, which are formally represented by completely positive trace-preserving (CPTP) maps. Such transformations can always be written as a unitary $U_{x_ky_k}^{(k)}$ applied jointly to the local system and some environment state $\xi_{\mE_k}$, with the environment being subsequently traced out. 
Formally, we have
\begin{equation}\label{general_mapping}
\mathcal{T}^{(k)}_{x_ky_k}:\sigma\mapsto\trr{{\mE_k}}{{U_{x_ky_k}^{(k)}} (\sigma\otimes\xi_{\mE_k} )U_{x_ky_k}^{(k)\dagger}}.
\end{equation}
We denote the Hilbert space for each party by $\mS_k$ and the total Hilbert space of the parties by $\mS=\bigotimes_{k=1}^{n}\mS_k$. 
Similarly, we denote the local and total Hilbert spaces of the environment by $\mE_k$ and $\mE=\bigotimes_{k=1}^{n}\mE_k$, respectively. 
The total initial environment state is thus $\xi_{\mE}=\bigotimes_{k=1}^n\xi_{\mE_k}$. 
(Note that in our SDI scheme we only assume a bound on the dimension of the output space of each party so, \emph{a priori}, this may be different from that of the input spaces so that the $U^{(k)}_{x_k y_k}$ instead map $\mS_k^i\otimes \mE_k^i \to \mS_k^f\otimes \mE_k^f$. 
For simplicity we assume the input and output spaces have the same dimensions in the proof below; the argument generalises trivially to the more general case.) 

Evaluating explicitly $\mathcal{A}_{n,d}$, we have
\begin{align}\label{sup1}
	\mathcal{A}_{n,d}&=\frac{1}{d^{2n}}\sum_{\substack{\mathbf{x},\mathbf{y},\mathbf{b}:\\\mathbf{b}=C(\mathbf{x},\mathbf{y})}} \trr{\mS}{\left(\bigotimes_{k=1}^{n} \mathcal{T}^{(k)}_{x_ky_k}\right)\!\left[\rho\right]\cdot M_\mathbf{b}}\notag\\
	&=\frac{1}{d^{2n}}\sum_{\substack{\mathbf{x},\mathbf{y},\mathbf{b}:\\\mathbf{b}=C(\mathbf{x},\mathbf{y})}} \trr{\mS\mE}{\left(\bigotimes_{k=1}^{n}U_{x_ky_k}^{(k)} \right)\left(\rho\otimes\xi_{\mE} \right)\left(\bigotimes_{k=1}^{n} U_{x_ky_k}^{(k)\dagger} \right)\left(M_\mathbf{b}\otimes\id_{\mE_k} \right)}\notag\\
	&=\frac{1}{d^{2n}} \sum_{\substack{\mathbf{x},\mathbf{y},\mathbf{b}\in\text{SEP}:\\\mathbf{b}=C(\mathbf{x},\mathbf{y})}}\trr{{\mS\mE}}{\left(\rho\otimes\xi_{\mE}\right) \bigotimes_{k=1}^{n}\left(U_{x_ky_k}^{(k)\dagger} \left(M^{(k)}_\mathbf{b}\otimes\id_{\mE_k}\right) U_{x_ky_k}^{(k)}\right) }\notag\\
	& \qquad\qquad\qquad + \frac{1}{d^{2n}} \sum_{\substack{\mathbf{x},\mathbf{y},\mathbf{b}\notin\text{SEP}:\\\mathbf{b}=C(\mathbf{x},\mathbf{y})}}\trr{{\mS\mE}}{\left(\rho\otimes\xi_{\mE}\right) \left(\bigotimes_{k=1}^{n} U_{x_ky_k}^{(k)\dagger}\right) \left(M_\mathbf{b}\otimes\id_{\mE_k}\right) \left( \bigotimes_{k=1}^{n}U_{x_ky_k}^{(k)}\right) }.		
\end{align}
Restricting ourselves to the first term above, we have 
\begin{align}\label{eq:T}
T& \equiv \frac{1}{d^{2n}} \sum_{\substack{\mathbf{x},\mathbf{y},\mathbf{b}\in\text{SEP}:\\\mathbf{b}=C(\mathbf{x},\mathbf{y})}}\trr{{\mS\mE}}{\left(\rho\otimes\xi_{\mE}\right) \bigotimes_{k=1}^{n}\left(U_{x_ky_k}^{(k)\dagger} \left(M^{(k)}_\mathbf{b}\otimes\id_{\mE_k}\right) U_{x_ky_k}^{(k)}\right) }\notag\\
& \le \frac{1}{d^{2n}}\sum_{\substack{\mathbf{x},\mathbf{y},\mathbf{b}\in\text{SEP}:\\\mathbf{b}=C(\mathbf{x},\mathbf{y})}} \trr{{\mS_1\mE_1}}{\left(\rho_{\mS_1}\otimes\xi_{\mE_1}\right)\left(U_{x_1y_1}^{(1)\dagger} \left(M^{(1)}_b\otimes\id_{\mE_1}\right) U_{x_1y_1}^{(1)}\right)}\prod_{k=2}^n\lambda_\text{max}[M^{(k)}_\mathbf{b}],
\end{align}
where $\rho_{\mS_1}=\trr{\mS_2\dots\mS_n}{\rho_{\mS}}$ and we have used Lemma~\ref{lemma} $n-1$ times, together with the identity $\lambda_\text{max}[U_{x_ky_k}^{(k)\dagger}( M^{(k)}_\mathbf{b}\otimes\id_{\mE_k}) U_{x_ky_k}^{(k)}]=\lambda_\text{max}[M^{(k)}_\mathbf{b}]$.

Since $\{M_\mathbf{b}\}_\mathbf{b}$ is a valid POVM it must satisfy $\sum_{\mathbf{b}}M_{\mathbf{b}}=\sum_{\mathbf{b}\in\text{SEP}}\bigotimes_{k=1}^n M^{(k)}_\mathbf{b} + \sum_{\mathbf{b}\notin\text{SEP}}M_{\mathbf{b}}=\id_\mS$. 
Tracing out subsystems $\{2,\ldots,n\}$, we see that 
\begin{equation}\label{part_trace}
\sum_{\mathbf{b}\in\text{SEP}} M^{(1)}_\mathbf{b}\prod_{k=2}^n \tr{M_\mathbf{b}^{(k)}}=d^{n-1}\id_{\mS_1} - \sum_{\mathbf{b}\notin\text{SEP}}\trr{\mS_2\cdots \mS_n}{M_{\mathbf{b}}}.
\end{equation}

By noting that, given the values of $y_1, x_1, \dots, x_{n-1}, \mathbf{b}$, the condition $\mathbf{b}=C(\mathbf{x},\mathbf{y})$ fixes the values of the remaining variables and that these remaining variables do not appear in the summand in Eq.~\eqref{eq:T}, we can rewrite the summation simply over $y_1, x_1, \dots, x_{n-1}, \mathbf{b}$.
Noting also that $0\le \lambda_\text{max}[P]\leq\tr{P}$ for any positive semidefinite operator $P$ and using Eq.~\eqref{part_trace} we have
\begin{align}\label{eq:Tbound}
	T&\leq \frac{1}{d^{2n}}\sum_{\substack{y_1,x_1,\dots, x_{n-1},\\\mathbf{b}\in\text{SEP}}} \trr{{\mS_1\mE_1}}{\left(\rho_{\mS_1}\otimes\xi_{\mE_1}\right)\left(U_{x_1y_1}^{(1)\dagger} \left(M^{(1)}_\mathbf{b}\prod_{k=2}^n\lambda_\text{max}[M^{(k)}_\mathbf{b}]\otimes\id_{\mE_1}\right) U_{x_1y_1}^{(1)}\right)}\notag\\
&\le \frac{1}{d^{2n}}\sum_{y_1,x_1,\dots, x_{n-1}} \trr{{\mS_1\mE_1}}{\left(\rho_{\mS_1}\otimes\xi_{\mE_1}\right)\left(U_{x_1y_1}^{(1)\dagger} \left(\sum_{\mathbf{b}\in\text{SEP}} M^{(1)}_\mathbf{b}\prod_{k=2}^n\tr{M^{(k)}_\mathbf{b}}\otimes\id_{\mE_1}\right) U_{x_1y_1}^{(1)}\right)}\notag\\
&= \frac{1}{d^{2n}}\sum_{y_1,x_1,\dots, x_{n-1}} \trr{{\mS_1\mE_1}}{\left(\rho_{\mS_1}\otimes\xi_{\mE_1}\right)\left(U_{x_1y_1}^{(1)\dagger} \left(d^{n-1}\id_{\mS_1\mE_1} - \sum_{\mathbf{b}\notin\text{SEP}}\trr{\mS_2\cdots \mS_n}{M_{\mathbf{b}}}\otimes\id_{\mE_1} \right) U_{x_1y_1}^{(1)}\right)}\notag\\
&\le \frac{1}{d} - \frac{1}{d^{2n}}\lambda_{\text{max}}\left[ \sum_{y_1,x_1,\dots, x_{n-1}}U_{x_1y_1}^{(1)\dagger} \left(\sum_{\mathbf{b}\notin\text{SEP}}\trr{\mS_2\cdots \mS_n}{M_{\mathbf{b}}}\otimes\id_{\mE_1} \right) U_{x_1y_1}^{(1)} \right].
\end{align}

We note briefly that one also obtains \eqref{eq:Tbound} if one considers general separable operators of the form $M_\mathbf{b} = \sum_i \bigotimes_{k=1}^n M^{(k)}_{\mathbf{b},i}$ (rather than simple tensor products).
This follows readily from the linearity of both Eqs.~\eqref{eq:T} and \eqref{part_trace}, so that the sum over separable measurement operators can be eliminated in the same way as above.

Continuing with the proof, note that for a positive semidefinite operator $P$ in a $d$-dimensional Hilbert space, $\lambda_{\text{max}}\left[P\right]\geq\frac{1}{d}\tr{P}$. Hence, letting $D$ be the dimension of $\mE_1$, we have
\begin{align}
T&\leq \frac{1}{d} - \frac{1}{Dd^{2n+1}}\trr{{\mS_1\mE_1}}{ \sum_{y_1,x_1,\dots, x_{n-1}}U_{x_1y_1}^{(1)\dagger} \left(\sum_{\mathbf{b}\notin\text{SEP}}\trr{\mS_2\cdots \mS_n}{M_{\mathbf{b}}}\otimes\id_{\mE_1} \right) U_{x_1y_1}^{(1)}}\notag\\
&= \frac{1}{d} - \frac{1}{Dd^{n+1}}\sum_{\mathbf{b}\notin\text{SEP}}\trr{{\mS\mE_1}}{M_{\mathbf{b}}\otimes\id_{\mE_1} }\notag\\
& \leq \frac{1}{d}-\frac{1}{d^{n+1}} \sum_{\mathbf{b}\notin\text{SEP}}\lambda_{\text{max}}\left[M_\mathbf{b}\right],
\end{align}
where, in the last step, we used the relation $\trr{\mS\mE_1}{M_\mathbf{b}\otimes\id_{\mE_A}}= D\Tr_\mS\left[M_\mathbf{b}\right]\geq D\lambda_{\text{max}}\left[M_\mathbf{b}\right]$.

Substituting this back into the expression~\eqref{sup1} for $\mathcal{A}_{n,d}$ and noting that 
\begin{equation}
	\trr{{\mS\mE}}{\left(\rho\otimes\xi_{\mE}\right) \left(\bigotimes_{k=1}^{n} U_{x_ky_k}^{(k)\dagger}\right) \left(M_\mathbf{b}\otimes\id_{\mE_k}\right) \left( \bigotimes_{k=1}^{n}U_{x_ky_k}^{(k)}\right) }
	\leq \lambda_{\text{max}} \left[M_\mathbf{b}\right]
	\leq 1,
\end{equation}
we have the bound
\begin{align}
	\mathcal{A}_{n,d} &\leq \frac{1}{d} + \left(\frac{1}{d^n}-\frac{1}{d^{n+1}}\right) \sum_{\mathbf{b}\notin\text{SEP}}\lambda_{\text{max}}\left[M_\mathbf{b}\right] \notag\\
	&\le \frac{1}{d} + \left(\frac{1}{d^n}-\frac{1}{d^{n+1}}\right)\left(d^n-k\right) \notag\\
	 &= \frac{1}{d^n}\left(d^n-k+\frac{k}{d}\right),
\end{align}
as desired.

\section{Partial tightness of Eq.~\eqref{measurement} for two parties}\label{D}

In this section we consider the tightness of the inequality~\eqref{measurement} for the bipartite case, which gives a lower bound on the number of entangled measurement operators compatible with a given average score $\mathcal{A}_{2,d}$. Specifically, we present a simple strategy that saturates the bound using a measurement for which there are $k=md$ separable measurement operators, for $m=0,\ldots,d$.

Consider the following projective joint measurement of two $d$-dimensional quantum systems, in which the measurement operators $M_\mathbf{b}=\ketbra{M_\mathbf{b}}{M_\mathbf{b}}$ are separable for all $(b_1,b_2)$ satisfying $b_2-b_1\in\{0,\ldots,m-1\}$, and entangled otherwise (where, as always, $b_2-b_1$ is computed modulo $d$). Hence, there are $md$ separable operators and $(d-m)d$ entangled operators. Specifically,
\begin{equation}\label{measD1}
	\ket{M_{b_1b_2}} =
	\begin{cases}
		\ket{b_1,b_2} & \text{ for $b_2-b_1\in\{0,\ldots,m-1\}$,}\\
		 Z^{b_1}\otimes X^{b_2-b_1}\ket{\phi_{\text{max}}} & \text{ for $b_2-b_1\notin\{0,\ldots,m-1\}$}.
	\end{cases}
\end{equation}
To see that this is a valid measurement, note that all the inner products of two different separable basis-elements is zero. Similarly, the entangled basis-elements constitute a subset of the Bell-basis, and are thus orthonormal. To show that the inner products between the separable basis-elements and entangled basis-elements are all zero, consider the straightforward calculation
\begin{equation}
\bra{b_1',b_2'}Z^{b_1}\otimes X^{b_2-b_1}\ket{\phi_{\text{max}}}=\frac{1}{\sqrt{d}}\sum_{\ell=0}^{d-1}\omega^{\ell b_1}\braket{b_1',b_2'}{\ell,\ell+b_2-b_1}=\frac{\omega^{b_1'b_1}}{\sqrt{d}}\delta_{b_2'-b_1',b_2-b_1}.
\end{equation} 
Since all the separable basis-elements have $b_2'-b_1'\in\{0,\ldots,m-1\}$ while all entangled basis-elements have $b_2-b_1\notin\{0,\ldots,m-1\}$, the final delta function is zero. Hence, Eq.~\eqref{measD1} defines an orthonormal basis.

Consider thus the following strategy. Let $A_1$ and $A_2$ share a maximally entangled state, and apply (a relabelled variant of ) the transformation strategy exploited several times in the main text, namely take the unitary transformations $U^{A_1}_{x_1y_1}=Z^{x_1}X^{y_1+x_1}$ and $U^{A_2}_{x_2y_2}=Z^{x_2}X^{y_2-x_2}$. It follows straightforwardly that 
\begin{equation}
\mathcal{A}_{2,d}={\frac{d-m}{d}+\frac{m}{d^2}}\,,
\end{equation}
which saturates the upper bound of Eq.~\eqref{measurement} for any $m$.
We leave it as an open question whether a similar partial tightness result holds in the more general $n$ partite case.

\end{document}